\documentclass[11pt]{article}

\usepackage{amsthm}
\usepackage{graphicx} 
\usepackage{array} 

\usepackage{amsmath, amssymb, amsfonts, verbatim}
\usepackage{hyphenat,epsfig,subcaption,multirow}
\usepackage[font=small,labelfont=bf]{caption}
\usepackage{bm}

\usepackage[usenames,dvipsnames]{xcolor}
\usepackage[ruled]{algorithm2e}

\usepackage{tcolorbox}
\tcbuselibrary{skins,breakable}
\tcbset{enhanced jigsaw}

\usepackage[normalem]{ulem}
\usepackage[compact]{titlesec}

\definecolor{DarkRed}{rgb}{0.5,0.1,0.1}
\definecolor{DarkBlue}{rgb}{0.1,0.1,0.5}

\usepackage{nameref}
\definecolor{ForestGreen}{rgb}{0.1333,0.5451,0.1333}
\definecolor{Red}{rgb}{0.9,0,0}
\usepackage[linktocpage=true,
	pagebackref=true,colorlinks,
	linkcolor=DarkRed,citecolor=ForestGreen,
	bookmarks,bookmarksopen,bookmarksnumbered]
	{hyperref}
\usepackage[noabbrev,nameinlink]{cleveref}
\crefname{property}{property}{Property}
\creflabelformat{property}{(#1)#2#3}
\crefname{equation}{eq}{Eq}
\creflabelformat{equation}{(#1)#2#3}

\usepackage{bm}
\usepackage{url}
\usepackage{xspace}
\usepackage[mathscr]{euscript}

\usepackage{tikz}
\usetikzlibrary{arrows}
\usetikzlibrary{arrows.meta}
\usetikzlibrary{shapes}
\usetikzlibrary{backgrounds}
\usetikzlibrary{positioning}
\usetikzlibrary{decorations.markings}
\usetikzlibrary{decorations.pathreplacing} 
\usetikzlibrary{patterns}
\usetikzlibrary{calc}
\usetikzlibrary{fit}
\usetikzlibrary{decorations}

\usepackage[framemethod=TikZ]{mdframed}

\usepackage[noend]{algpseudocode}
\makeatletter
\def\BState{\State\hskip-\ALG@thistlm}
\makeatother

\usepackage{cite}
\usepackage{enumitem}
\setlist[itemize]{leftmargin=20pt}
\setlist[enumerate]{leftmargin=20pt}

\usepackage[margin=1in]{geometry}

\usepackage{thmtools}
\usepackage{thm-restate}

\newtheorem{theorem}{Theorem}
\newtheorem{lemma}{Lemma}[section]
\newtheorem{proposition}[lemma]{Proposition}
\newtheorem{corollary}[lemma]{Corollary}
\newtheorem{claim}[lemma]{Claim}
\newtheorem{fact}[lemma]{Fact}

\newtheorem*{claim*}{Claim}
\newtheorem*{assumption*}{Assumption}
\newtheorem*{proposition*}{Proposition}
\newtheorem*{lemma*}{Lemma}

\newtheorem{observation}[lemma]{Observation}

\newtheorem*{theorem*}{Theorem}

\crefname{lemma}{Lemma}{Lemmas}
\crefname{claim}{claim}{claims}
\crefname{property}{Property}{Properties}
\crefname{invariant}{Invariant}{Invariants}

\newtheorem{mdresult}{Result}
\newenvironment{result}{\begin{mdframed}[backgroundcolor=lightgray!40,topline=false,rightline=false,leftline=false,bottomline=false,innertopmargin=2pt]\begin{mdresult}}{\end{mdresult}\end{mdframed}}

\theoremstyle{definition}

\newtheorem{remark}[lemma]{Remark}
\newtheorem{problem}{Problem}
\newtheorem{definition}[lemma]{Definition}

\newenvironment{Problem}{\begin{mdframed}[topline=true,bottomline=true, innerbottommargin=5pt,innertopmargin=5pt]\begin{problem}}{\end{problem}\end{mdframed}}

\newtheorem*{mdproblem*}{Problem}
\newenvironment{Problem*}{\begin{mdframed}[hidealllines=false,innerleftmargin=10pt,backgroundcolor=gray!10,innertopmargin=5pt,innerbottommargin=5pt,roundcorner=10pt]\begin{mdproblem*}}{\end{mdproblem*}\end{mdframed}}
\newtheorem{mddefinition}[lemma]{Definition}

\newtheorem*{mddefinition*}{Definition}
\newenvironment{Definition*}{\begin{mdframed}[hidealllines=false,innerleftmargin=10pt,backgroundcolor=white!10,innertopmargin=5pt,innerbottommargin=5pt,roundcorner=10pt]\begin{mddefinition*}}{\end{mddefinition*}\end{mdframed}}
\newtheorem{mdremark}{Remark}

\newenvironment{ourbox}{\begin{mdframed}[hidealllines=false,innerleftmargin=10pt,backgroundcolor=white!10,innertopmargin=2pt,innerbottommargin=5pt,roundcorner=10pt]}{\end{mdframed}}

\newtheorem{mdalgorithm}{Algorithm}

\newenvironment{proofof}[1]{{\medbreak\noindent \em Proof of #1.  }}{\hfill\qed\medbreak}

\allowdisplaybreaks

\renewcommand{\qed}{\nobreak \ifvmode \relax \else
      \ifdim\lastskip<1.5em \hskip-\lastskip
      \hskip1.5em plus0em minus0.5em \fi \nobreak
      \vrule height0.75em width0.5em depth0.25em\fi}

\renewcommand{\leq}{\leqslant}
\renewcommand{\geq}{\geqslant}

\newcommand{\rs}{{{Ruzsa-Szemerédi}}\xspace}
\newcommand{\st}{\textnormal{subject to}}

\newcommand{\Ot}{\ensuremath{\widetilde{O}}}
\newcommand{\eps}{\ensuremath{\varepsilon}}

\newcommand{\bracket}[1]{\left[#1\right]}
\newcommand{\paren}[1]{\ensuremath{\left(#1\right)}\xspace}
\newcommand{\card}[1]{\left\vert{#1}\right\vert}

\newcommand{\norm}[1]{\ensuremath{\|#1\|}}

\newcommand{\ceil}[1]{{\left\lceil{#1}\right\rceil}}

\newcommand{\set}[1]{\ensuremath{\left\{ #1 \right\}}}
\newcommand{\poly}{\mbox{\rm poly}}

\DeclareMathOperator*{\Exp}{\ensuremath{{\mathbb{E}}}}
\DeclareMathOperator*{\Prob}{\ensuremath{\textnormal{Pr}}}
\renewcommand{\Pr}{\Prob}

\newenvironment{tbox}{\begin{tcolorbox}[
		enlarge top by=5pt,
		enlarge bottom by=5pt,
		 breakable,
		 boxsep=0pt,
                  left=4pt,
                  right=4pt,
                  top=10pt,
                  arc=0pt,
                  boxrule=1pt,toprule=1pt,
                  colback=white
                  ]
	}
{\end{tcolorbox}}

\newcommand{\II}{\ensuremath{\mathbb{I}}}
\newcommand{\HH}{\ensuremath{\mathbb{H}}}

\newcommand{\mireal}[1][]{
  \ifx\relax#1\relax%
    \II(\mione \,; \mitwo)%
  \else%
    \II(\mione \,; \mitwo\mid #1)%
  \fi
}
\newcommand{\en}[1]{\ensuremath{\HH(#1)}}

\newcommand{\prot}{\ensuremath{\pi}}

\newcommand{\ignore}[1]{}

\newcommand{\defeq}{:=}
\newcommand{\load}{\ensuremath{\textsc{load}}}
\newcommand{\maxload}{\ensuremath{\textsc{MaxLoad}}}
\newcommand{\optload}{\ensuremath{\textsc{OPTload}}}
\newcommand{\mathP}{\mathcal{P}}

\newcommand{\A}{\mathcal{A}}
\newcommand{\Astar}{\mathcal{A}^*}

\newcommand{\spar}{\ensuremath{\textnormal{\textsc{sparsifier}}}}
\newcommand{\MC}{\ensuremath{\textnormal{\textsf{MC}}}}
\newcommand{\LP}{\ensuremath{\textnormal{\textsf{LP}}}}
\newcommand{\rsplus}{Matching-Contractor}

\title{Streaming and Communication Complexity of Load-Balancing\\ via Matching Contractors} 
\author{Sepehr Assadi\footnote{(sepehr@assadi.info) Supported in part by a  Sloan Research Fellowship, an NSERC
Discovery Grant, a University of Waterloo startup grant, and a Faculty of Math Research Chair grant. \smallskip} 
\\ {\small ~~University of Waterloo~~} \and
Aaron Bernstein\footnote{(bernstei@gmail.com) Supported in part by a  Sloan Research Fellowship, a Google Research Fellowship, NSF Grant 1942010, and a Charles S. Baylis endowment from NYU. \smallskip} \\ {\small ~~New York University~~} \and 
Zachary Langley \\ {\small ~~Rutgers University~~} \and 
Lap Chi Lau\footnote{(lapchi@uwaterloo.ca) Supported by an NSERC Discovery Grant. \smallskip} \\ {\small ~~University of Waterloo~~} \and
Robert Wang \\ {\small ~~University of Waterloo~~}
}

\date{}

\begin{document}

\maketitle


\begin{abstract}

In the load-balancing problem, we have an $n$-vertex bipartite graph $G=(L, R, E)$ between a set of clients and servers. The goal is to find an assignment of all clients to the servers, while minimizing the maximum load on each server, 
where load of a server is the number of clients assigned to it. Motivated by understanding the streaming complexity of this problem, we study load-balancing in the one-way (two-party) communication model: 
the edges of the input graph are partitioned between Alice and Bob, and Alice needs to send a short message to Bob for him to output a solution of the entire graph. 

\medskip

We show that settling the one-way communication complexity of load-balancing is equivalent to a natural sparsification problem for load-balancing, which can alternatively be interpreted as sparsification for vertex-expansion.
We then prove a \emph{dual} interpretation of this sparsifier, showing that the minimum density of a sparsifier is 
effectively the same as the maximum density one can achieve for an extremal graph family that is new to this paper, called \emph{\rsplus s}; these
graphs are intimately connected to the well-known \rs graphs and generalize them in certain aspects. Our chain of equivalences thus shows that the one-way communication complexity of load-balancing can be reduced to a purely graph theoretic question: what is the maximum density of a \rsplus\ on $n$ vertices?

\medskip

As our final result, we present a novel combinatorial construction of some-what dense \rsplus s, which implies 
a strong one-way communication lower bound for load-balancing: any one-way protocol (even randomized) with $\Ot(n)$ communication cannot achieve a better than $n^{\frac14-o(1)}$-approximation. Previously, no non-trivial lower bounds 
were known for protocols with even $O(n\log{n})$ bits of communication (a better-than 2-approximation lower bound is trivial). Our result also implies the first non-trivial lower bounds for semi-streaming 
load-balancing in the edge-arrival model, ruling out $n^{\frac14-o(1)}$-approximation in a single-pass. 
\end{abstract}

\pagenumbering{roman}

\clearpage

\setcounter{tocdepth}{3}
\tableofcontents
\clearpage
\pagenumbering{arabic}
\setcounter{page}{1}

\clearpage


\newcommand{\LB}{\ensuremath{\textnormal{\textsf{LoadBal}}}}

\newcommand{\randC}[1]{\ensuremath{\vec{R}(#1)}}

\renewcommand{\defeq}{:=}

\section{Introduction}\label{sec:intro}

We study the \textbf{load-balancing} problem. Given a bipartite graph $G=(L, R,E)$, an \emph{assignment} maps each vertex in $L$ to one of its neighbors in $R$. The load of vertex in $R$ is the number of vertices in $L$ assigned to it. The goal is to find an assignment that minimizes the maximum load. We often refer to vertices in $L$ as \emph{clients} and the ones in $R$ as \emph{servers}.

Load-balancing has a rich history under different names. It is studied in the scheduling literature as job scheduling with restricted assignment~\cite{Horn73,BrunoCS74,LenstraST90,LinL04,HarveyLLT06,JansenR17,JansenR20}, in the distributed computing literature as backup placement problem~\cite{CzygrinowHSW12,HalldorssonKPR18,OrenBL18,BarenboimO20,AssadiBL20,AhmadianLPZ21}, and 
in graph algorithms as the semi-matching problem~\cite{HarveyLLT06,KonradR13,FakcharoenpholLN14}. 
From an optimization perspective, it serves as a canonical example of a mixed packing-covering problem, with packing constraints on the servers and covering constraints on the clients. 

This work focuses on the load-balancing problem in the semi-streaming model~\cite{FeigenbaumKMSZ05}, where the edges of the input graph $G$ arrive one-by-one in a stream, and the algorithm is allowed to
process these edges using $O(n \cdot \poly\!\log{\!(n)})$ memory, where $n$ is the total number of vertices. In addition, the algorithm is limited to a single pass (or a few passes) over the stream. 
These constraints capture several challenges of processing massive graphs, including I/O-efficiency and efficiently monitoring evolving graphs. 
As such, the semi-streaming model has been at the forefront of the research on massive graphs in recent years.

Load-balancing can be seen as a natural and useful problem between two of the most well-studied families of problems in the semi-streaming model: 
matching problems 
(see e.g.~\cite{FeigenbaumKMSZ05,GoelKK12,AhnG11,AssadiKLY16,FischerMU22,Assadi24})
and coverage problems
(see e.g.~\cite{DemaineIMV14,AssadiKL16,McGregorV17,KhannaKA23}).
Yet, in sharp contrast to these problems, 
our understanding of semi-streaming load-balancing is quite limited. 
Over a decade ago, Konrad and Ros{\'{e}}n~\cite{KonradR13} initiated the study of load-balancing in this model,
presenting a simple $O(\sqrt{n})$-approximation algorithm in a single-pass and an $O(\log{n})$-approximation algorithm in $O(\log{n})$ passes. 
The latter algorithm was only recently improved to an $O(1)$-approximation in $O(\log{n})$ passes~\cite{AssadiBL23}. 
At this point, there is no evidence why even a $2$-approximation cannot be achieved in a single pass!\footnote{Obtaining a better-than-two approximation
for load-balancing requires finding a perfect matching when it exists. As such, one can borrow existing lower bounds for the latter problem 
to obtain that $\Omega(\log{n})$ passes are required for load-balancing in this case.}

\subsection{Our Contributions}\label{sec:contributions}
 
A main consequence of this work is that Konrad and Ros\'en's single-pass $O(\sqrt{n})$-approximation algorithm cannot be significantly improved.

\begin{result}\label{res:streaming}
There is no semi-streaming algorithm that obtains a $n^{\frac14-o(1)}$-approximation to the load-balancing problem with success probability at least $\frac23$.
\end{result}

This result stems from a series of reductions and equivalences that we describe below.
Ultimately, we reduce the problem to a question in extremal graph theory and present a novel construction to establish the streaming complexity lower bound.

\subsubsection*{One-Way Communication Complexity and Load-Balancing Sparsifiers}

Inspired by previous work on matchings and coverage problems~\cite{GoelKK12,KonradR13,AssadiKL16,FeldmanNSZ20}, 
we study the {\em one-way communication complexity} of the load balancing problem.
In this communication model, the input graph is edge-partitioned between Alice and Bob. Alice sends a single $\Ot(n)$-bit message to Bob, and Bob then outputs an approximately optimal solution to the entire graph. 
Semi-streaming algorithms imply one-way communication protocols (but not vice versa), so a lower bound on the one-way communication complexity translates to a lower bound on the semi-streaming complexity.

Our first equivalence shows that the one-way communication complexity of load-balancing is nearly equal to the minimum density of {\bf load-balancing sparsifiers},
which were called {\em semi-matching skeletons} in~\cite{KonradR13}. 
For a graph $G=(L,R,E)$ and approximation ratio $\alpha \geq 1$, 
we define an $\alpha$-approximate load-balancing sparsifier of $G$ 
to be any spanning subgraph $H$ of $G$ that can preserve the value of optimal load-balancing for \emph{every subset} of clients in $L$ up to a factor of $\alpha$ (see \autoref{dfn:lb-sparsifier} for the formal definition). 

The original motivation of this definition in~\cite{KonradR13} is that a good load-balancing sparsifier implies a good one-way communication protocol.
The protocol is the natural one: Alice sends Bob an $\alpha$-approximate sparsifier $H_A$ of her graph $G_A$, and Bob outputs the optimal assignment $\A$ of the graph $H_A \cup G_B$. It is not hard to check that $\A$ is an $(\alpha+1)$-approximation to the optimal assignment in $G_A \cup G_B$ (see \Cref{thm:spar-implies-protocol}).

Our contribution is proving the other direction of the equivalence, showing that load-balancing sparsifiers are the correct combinatorial objects for understanding the one-way communication complexity.

\begin{result}\label{res:equivalence}
	For any $\alpha \geq 1$ and large $n \geq 1$, the communication cost of the best one-way protocol for $\alpha$-approximate load-balancing, possibly with randomization and success probability $\frac23$, 
	is equal, up to $\poly\log{(n)}$ factors, to the smallest $T$ such that any $n$-vertex graph contains a $\Theta(\alpha)$-approximate load-balancing sparsifier with at most $T$ edges. 
\end{result}

\subsubsection*{Load-Balancing Sparsifiers and Matching Contractors}

Our key conceptual contribution is the identification of a natural and interesting dual object for load-balancing sparsifiers. 
We say a bipartite graph $G=(L,R,E)$ is an \textbf{$\bm{\alpha}$-\rsplus} if and only if its edges can be partitioned into
a set of matchings $M_1,M_2,\ldots,M_k$ with the following property: vertex-set of each matching $M_i$ in $L$ is ``heavily contracting'' if we do not use edges of $M_i$ itself; more formally, the neighbor-set of $L(M_i)$ in $G \setminus M_i$ has size $\card{M_i}/\alpha$ only. 
\rsplus s can be roughly thought of as a more stringent version of \emph{\rs graphs}~\cite{RuzsaS78}: see \Cref{comparison} 
for a more detailed comparison. 
Our second equivalence is between the sparsity of a load-balancing sparsifier and the density of a \rsplus.

\begin{result}\label{res:load-balancing-contractor}
	For any $\alpha \geq 1$ and large $n \geq 1$, 
	the minimum $T$ such that every graph $n$-vertex graph $G$ contains an $\alpha$-approximate load-balancing sparsifier with at most $T$ edges, is equal, up to $\poly\log{(n)}$ factors, to the 
	largest density of a $\Theta(\alpha)$-\rsplus\ with $n$ vertices.  
\end{result}

\subsubsection*{Construction of Dense \rsplus s} 

\Cref{res:equivalence} and \Cref{res:load-balancing-contractor} reduce the one-way communication complexity of load balancing to a new question in extremal graph theory: what is the maximum density of a \rsplus\ on $n$ vertices?
Our key technical contribution is a novel construction of somewhat-dense \rsplus s.

\clearpage

\begin{result}\label{res:contractor-exists}
	For any sufficiently small $\eps > 0$ and large integer $n \geq 1$, there exists an $\big( n^{\frac14-O(\eps)} \big)$-\rsplus\ with  $n^{1+\Omega(\eps^2)}$ edges. 
\end{result}

Combining our results together implies that no $\Ot(n)$-communication protocol can achieve $n^{\frac14-\Omega(1)}$-approximation to load-balancing. 
This in turn implies \autoref{res:streaming} that the best approximation ratio achievable by single-pass semi-streaming algorithms is only $n^{\frac14-o(1)}$, significantly improving the prior $2$-approximation lower bounds. 
Our lower bound also comes quite close to the $n^{\frac13}$ one-way communication complexity upper bound of \cite{KonradR13}, and by the machinery developed in this paper, closing this gap now amounts to pinning down the maximum density of \rsplus s.

\subsection{Previous Work}

Konrad and Ros\'en~\cite{KonradR13} made the first step in understanding the one-way communication complexity of load-balancing. 
They defined \emph{load-balancing sparsifiers} (called \emph{semi-matching skeletons} in~\cite{KonradR13}) and showed that the existence of sparse load-balancing sparsifiers implies good communication protocols.
They then use this connection to obtain an $n^{\frac13}$-approximate one-way protocol to load-balancing. 
Their result leaves two significant gaps. 
Firstly, while they showed a one-way relation between load-balancing sparsifiers and communication protocols, 
they did not prove an equivalence between the two,
hence leaving it uncertain whether this is the ``right" notion of sparsification. 
Secondly, their lower bounds for both sparsifiers and communication protocols are quite limited. They show that a $n^{\frac{1}{c+1}}$-approximate one-way protocol requires $cn$ bits, but this has no implication for protocols that communicate $O(n\log(n))$ bits, and hence no implication for semi-streaming in general.

Our work is heavily inspired by the pioneering work of~\cite{GoelKK12} for the closely related maximum matching problem.
\cite{GoelKK12} initiated a systematic study of semi-streaming matchings through one-way communication complexity.
The two main discoveries of~\cite{GoelKK12} are as follows.
\begin{itemize}
\item They introduced \emph{matching covers} as a natural notion of sparsifiers for matchings, and proved an equivalence between matching covers and  one-way communication complexity of matchings.
Matching covers have since found far reaching implications in streaming (e.g., in~\cite{Kapralov13,AssadiB19,Bernstein20,AssadiBKL23}) and beyond (e.g., in dynamic graph algorithms~\cite{AssadiBKL23,BehnezhadG24,AssadiK24}).

\item They formulated the natural ``dual'' connection between matching covers 
and \rs graphs~\cite{RuzsaS78,FischerLNRRS02,AlonMS12}, an extremal graph family with many \emph{large} edge-disjoint \emph{induced} matchings. 
They thus reduced the one-way communication problem to a problem in extremal graph theory: do there exist dense \rs graphs? 
Extending known constructions of \rs graphs in~\cite{FischerLNRRS02} for monotonicity testing lower bound,
they proved the first non-trivial lower bound for approximation of matchings in the semi-streaming model. 
Remarkably, after~\cite{GoelKK12} first established this connection, \rs graphs have become a central tool for proving semi-streaming lower bounds 
for matchings and beyond (e.g., in~\cite{Kapralov13,AssadiKLY16,AssadiKL17,AssadiR20,Kapralov21,ChenKPSSY21,AssadiS23,AssadiKNS24}).
\end{itemize}

\subsection{Roadmap and Technical Overview}

In this work, we follow a similar chain of reduction and equivalences as in~\cite{GoelKK12} to relate the one-way communication complexity of load-balancing to the density of \rsplus s. 
In the following, we provide an outline of this paper and discuss the high-level ideas in each step.

In \Cref{sec:prelim} we introduce basic notation and preliminaries. 
In \Cref{sec:sparsifiers}, we define the two crucial objects of this paper: load-balancing sparsifiers (introduced in \cite{KonradR13}) and \rsplus s (introduced in this work).
We should think of load-balancing sparsifiers as the analog of matching covers in~\cite{GoelKK12}, and \rsplus s as the analog of \rs graphs.
In \autoref{sec:equivalence-intuition}, we provide some intuition on the equivalence of load-balancing sparsifiers and one-way communication complexity; the proof is in \autoref{sec:cc-lb} using information theory and \rsplus s.
We show in \autoref{sec:equiv-notions-sparsification} that load-balancing sparsifiers are equivalent to a version of {\em vertex expansion sparsifiers} in bipartite graphs, suggesting that it is a natural object of its own interest. 
We also formulate an equivalent operational definition that is simpler to work with.
In \autoref{sec:rsplus}, we discuss the relation between \rsplus s and \rs graphs. 

In \Cref{sec:LP}, we prove \Cref{res:load-balancing-contractor} about the equivalence between the existence of sparse load-balancing sparsifiers and the non-existence of dense \rsplus s. 
This section encapsulates the main conceptual contribution of this paper, 
that \rsplus s is the dual object of load-balancing sparsifiers in a precise sense.
To prove this, we formulate a linear programming relaxation for finding a load-balancing sparsifier of a bipartite graph, and prove that it is an $O(\log n)$-approximation.
Then we construct the dual linear program and prove that one can construct a \rsplus\ from a solution to the dual LP.
This equivalence is the analog of the equivalence between matching covers and \rs graphs in~\cite{GoelKK12}, and their proof is also based on linear programming duality.
For our proof, both steps are done by a randomized rounding argument using some problem-specific insights.

In \Cref{sec:construction}, we show an explicit construction of \rsplus s (\Cref{res:contractor-exists}), 
which is the main technical challenge in this paper.
One reason is that the existence of dense \rsplus s is counter-intuitive:  
For the complete bipartite graphs, one can prove that there are load-balancing sparsifiers with $O(n \log n)$ edges by random sampling.
These are known as magical graphs in~\cite{HLW06} and have applications in error-correcting codes and super-concentrators.
Given this and the recent successes of sparsification in various settings, our initial effort was to prove that sparse load-balancing sparsifiers always exist (and hence dense \rsplus s do not).
Another challenge is that unlike for \rs graphs, there were no known constructions of \rsplus s, and hence no clear starting point in terms of what tools to use.
As we discuss in \autoref{sec:rsplus}, \rsplus\ is an even more stringent version of \rs graph, which is itself notoriously difficult to construct~\cite{FischerLNRRS02,AlonMS12}. 
Our construction idea is to view each vertex as a string, and use the block structures of the vertices and a set family of small pairwise intersection to argue about the contraction property.
Even though the final construction and the analysis are short and elementary, we see this as the key technical innovation in this paper.

Finally, in \Cref{sec:cc-lb}, we show the equivalence between the communication complexity of load-balancing and the existence of dense \rsplus s.
Combining with the equivalence in \Cref{res:load-balancing-contractor}, this completes \Cref{res:equivalence}.
This equivalence is the analog of the equivalence between the communication complexity of matchings and the existence of dense \rs graphs in~\cite{GoelKK12}. 
Our proof uses basic information theoretic arguments and a simple modification of \rsplus s to establish the lower bound.

Overall, we find it interesting to have natural analogs of matching covers and \rs graphs in load-balancing sparsifiers and \rsplus s.
Given the large impact of~\cite{GoelKK12}, we hope that these new objects and equivalences can play a further role in understanding load-balancing in streaming and other models. 
One promising direction is to establish lower bound on multi-pass semi-streaming algorithms for the load balancing problem. 
We believe they also have the potential to shed light on other graph problems, such as vertex expansion and matching conductance as discussed in \autoref{sec:equiv-notions-sparsification}.


\section{Preliminaries}\label{sec:prelim}

\paragraph{Notation.} 
Given two functions $f, g$, we use $f \lesssim g$ to denote the existence of a positive constant $c > 0$, such that $f \leq c \cdot g$ always holds.
We use $f \asymp g$ to denote $f \lesssim g$ and $g \lesssim f$.

\paragraph{Graphs.} 
We use $G=(L,R,E)$ to denote a bipartite graph with sides $L$ and $R$. 
For any vertex $v$, we define $N_G(v)$ to be the set of neighbors of $v$ in $G$, $E_G(v)$ to be the set of incident edges,  $N_G(S) \defeq \bigcup_{v \in S} N(v)$ and $E_G(S) \defeq \bigcup_{v \in S} E(v)$. Given subsets $X \subseteq L$ and $Y \subseteq R$, we define $E_G(X,Y) \defeq \{(u,v) \in E \mid u \in X \land v \in Y\}$ and we define $G[X \cup Y]$ to be the induced graph $(X \cup Y, E_G(X,Y))$. When clear from the context, we may drop the subscript $G$ in these notation. 

We define a matching $M \subseteq E$ to be a set of disjoint edges. We let $L(M)$ denote the left endpoints of edges in $M$, and $R(M)$ the right endpoints. For any graph $G$, we let $\mu(G)$ denote the size of the maximum matching in $G$. Finally, we say that a set $X \subseteq V$ is \textbf{matchable} if there exists a matching for which every vertex in $X$ is matched.

\paragraph{Load Balancing.} In the context of load balancing, we will often refer to vertices in $L$ as \emph{clients} and vertices in $R$ as \emph{servers}.  An assignment $\A$ assigns every client to some server: formally, $\A$ is a function $\A: L \rightarrow R$ such that for any client $c \in L$, we have $\A(c) \in N(c)$. For any server $s \in R$, we define $\A^{-1}(s) \defeq \{c \in L \mid \A(c) = s\}$, and we define the load of a server to be $\load_\A(s) \defeq |\A^{-1}(s)|$. 

We define $\maxload(A) \defeq \max_{s \in R}\load(s)$. The goal of the load-balancing problem is to find an assignment $\A$ that minimizes $\maxload(\A)$. To this end, we define $\optload(G)$ to be the load of the optimal assignment on $G$. 
We say that assignment $\A$ for $G$ is $\alpha$-approximate if $\maxload(\A) \leq \alpha \cdot \optload(G)$.

We use the following generalization of Hall's Theorem~\cite{Hall87}, proved in \cite[Lemma 4]{KonradR13arxiv}.

\begin{proposition}[\cite{KonradR13arxiv}]
	\label{lem:lb-halls}
For any bipartite graph $G = (L, R, E)$, 
\[
\optload(G) = \max_{\emptyset \neq X \subseteq L} \ceil{\frac{|X|}{|N(X)|}}.
\]
\end{proposition}

\subsection{One-Way Communication Complexity}\label{sec:cc}

We work with the standard one-way two-player communication complexity model of Yao~\cite{Yao79} (see the excellent textbooks by~\cite{KNbook} and~\cite{RYbook} for the background on this model). 
Specifically, we are interested in the following problem. 

\begin{Problem}\label{prob:one-way}
	An $n$-vertex bipartite graph $G=(C, S, E)$ with bipartition $C$ of \emph{clients} and $S$ of \emph{servers} is edge-partitioned (arbitrarily) between two players: Alice receives $E_A \subseteq E$ and Bob receives $E_B \subseteq E$, 
	where $E_A \cup E_B = E$ and $E_A \cap E_B = \emptyset$.  
	
	The goal is for Alice to send a single message to Bob which is only a function of her input $G_A := (C, S, E_A)$, and Bob, given this message and his input $G_B := (C, S, E_B)$ should 
	output a solution to the load-balancing problem on the entire graph $G$. In a randomized protocol, we assume Alice and Bob also have access to a \emph{shared} source of randomness, commonly 
	referred to as \textbf{public randomness}. In that case, the message of Alice and the output of Bob can also additionally depend on this public randomness.
	
	For any $n \geq 1$ and approximation ratio $\alpha \geq 1$, we use $\LB(n,\alpha)$ to denote this problem on $n$-vertex graphs wherein the goal is to obtain (at least) an $\alpha$-approximate solution. 
\end{Problem}

We refer to the algorithm that decide the messages of Alice and the output of Bob in~\Cref{prob:one-way} as a \textbf{protocol} $\pi$. 
The main measure of interest for us is the \textbf{communication cost} of a protocol $\pi$, denoted by $\norm{\pi}$, and defined as the worst-case length of the message
Alice sends to Bob (without loss of generality, via a padding argument, we assume length of all the messages communicated in the protocol is the same). 
Finally, we define the \textbf{(randomized) communication complexity} of $\LB(n,\alpha)$ as the minimum communication cost of any randomized protocol that solves this problem with probability of success at least $\frac23$, denoted by $\randC{\LB(n,\alpha)}$.



\section{Load-Balancing Sparsifiers}\label{sec:sparsifiers}

One of the key contributions of our paper is showing that the one-way communication problem of \autoref{sec:cc} is nearly equivalent to a notion of load-balancing sparsifiers, first introduced by Konrad and Ros\'en \cite{KonradR13}. 
Given a graph $G = (L, R, E)$, a load-balancing sparsifier is a subgraph $H = (L, R, E_H)$ that approximately preserves all the load-balancing properties of $G$. 

\begin{definition}[\!\!\cite{KonradR13}] \label{dfn:lb-sparsifier}
Given $G = (L, R, E)$ and $\alpha \geq 1$, we say that subgraph $H = (L, R, E_H)$ is an \textbf{$\bm{\alpha}$-approximate load-balancing sparsifier} of $G$ iff for every set $C \subseteq L$, 
\[
\optload(H[C \cup R]) \leq \alpha \cdot \optload(G[C \cup R]).
\] 
When the context is clear, we sometimes refer to $H$ as simply a $\alpha$-sparsifier of $G$. 
\end{definition}

Just like with other sparsification problems, the natural question is whether every graph $G$ contain an $\alpha$-approximate load-balancing sparsifier with few edges.

\begin{definition} \label{dfn:span} \label{dfn:spar}
Define:
\begin{itemize}
\item $\spar(G,\alpha)$ to be the minimum possible number of edges in an $\alpha$-approximate load-balancing sparsifier of $G$;
\item $\spar(n,\alpha)$ to be the maximum $\spar(G,\alpha)$ over all bipartite graphs $G = (L, R, E)$ with $\card{L} = n$. 
\end{itemize}
Note that $\spar(n,\alpha)$ and $\spar(G,\alpha)$ are monotonically decreasing as $\alpha$ increases. 
\end{definition}

Konrad and Ros{\'{e}}n~\cite{KonradR13}, who referred to these sparsifiers as {\em semi-matching skeletons}, presented the following nontrivial upper bound.

\begin{proposition}[\!\!\cite{KonradR13}]
\label{thm:sparsifier-ub}
 Every graph $G$ with $n$ clients contains a $n^{\frac13}$-sparsifier with at most $2n$ edges, i.e.
\[\spar(n,n^{\frac13}) \leq 2n.\]
\end{proposition}	

They also showed lower bounds of (roughly) the form  $\spar(n,n^{\frac{1}{c+1}}) \geq cn$ for $c \geq 1$. These lower bounds however have no implications for sparsifiers with $O(n\log{n})$ edges. 
In particular, the possibility that every graph $G$ contains an $O(1)$-sparsifier with $O(n\log n)$ edges was not ruled out.
Note that when $G$ is a complete bipartite graph, one can indeed construct a $O(1)$-sparsifier with $O(n \log n)$ edges by random sampling (see e.g.~the section about magical graphs in~\cite{HLW06}).

\subsection{Equivalence Between Sparsifier and One-Way Communication Complexity}
\label{sec:equivalence-intuition}

Konrad and Ros{\'{e}}n~\cite{KonradR13} already showed this equivalence in one direction.
We include a proof in \Cref{sec:app-proofs} for completeness.

\begin{proposition}[\!\!\cite{KonradR13}] \label{thm:spar-implies-protocol}
Fix $n,\alpha \geq 1$, and suppose $\spar(n,\alpha) = T$. 
Then, there exists a deterministic protocol $\pi$ for $\LB(n,\alpha+1)$ with communication cost $\norm{\pi} = O(T\log(n))$ bits.
\end{proposition}

Combining \autoref{thm:spar-implies-protocol} with their sparsifiers in \autoref{thm:sparsifier-ub} gives the following result.

\begin{corollary}[\!\!\cite{KonradR13}]
There exists a deterministic protocol $\pi$ for $\LB(n,n^{\frac13}+1)$ with communication cost $\norm{\pi} = O(n^{\frac13}\log(n))$ bits.
\end{corollary}

We show that this equivalence also goes in the other direction, even for randomized protocols. 

\begin{theorem} \label{thm:equivalence} 
Suppose there exists a (randomized) communication protocol $\pi$ for $\LB(n,\alpha)$ with communication cost $\norm{\pi} = C$ and probability of success at least $\frac23$. Then, 
\[
\spar(n,8\alpha) \lesssim C \cdot \log^2{(n)}.
\]
\end{theorem}

\paragraph{Intuition for Theorem 1:} The full proof will be presented in \Cref{sec:equivalence-proof} as it requires the new notion of \rsplus s in \autoref{sec:rsplus}, as well as the chain of equivalences worked out in Section \ref{sec:LP}.

For the sake of intuition, let us make the (small) assumptions that $|L| = |R| = n$ and that the protocol $\pi$ is deterministic, as well as the (large) assumption that $\pi$ takes the following form: Alice's message is limited to some subgraph $H_A$ of $G_A$. 
Clearly $H_A$ has $O(C)$ edges.
We argue that $H_A$ must be an $\alpha$-sparsifier of $G_A$. 
Suppose for contradiction that $H_A$ is not a $\alpha$-sparsifier. 
Then there must exist some set $X \subseteq L$ such that
\begin{equation}
\label{eq:bad-set}
\optload(H_A[X \cup R]) > \alpha \cdot \optload(G_A[X \cup R]) 
\end{equation}

Now, say that Bob's graph $G_B$ contains a complete graph from $L-X$ to $R$, but no edges incident to $X$. The key observation is that
\begin{equation}
\label{eq:spar-C-intuition}
\optload(H_A \cup G_B) = \optload(H_A[X \cup R]).
\end{equation}

We now justify both directions of Equation \ref{eq:spar-C-intuition}. It is easy to see that $\optload(H_A \cup G_B) \geq \optload(H_A[X \cup R])$, since $G_B$ contains no edges incident to $X$
To see that $\optload(H_A \cup G_B) \leq \optload(H_A[X \cup R])$, consider the optimal assignment $\A$ in $H_A[X \cup R])$. Let $S_\A$ be all servers that have load at least $1$ in $\A$; note that $\card{S_{\A}} \leq |X|$, which implies $\card{L \setminus X} \leq \card{R \setminus S_{\A}}$ (because we assumed $\card{L} = \card{R}$), so by construction of $G_B$, the set $L \setminus X$ is matchable in $G_B[(L \setminus X) \cup (R \setminus S_\A)]$. Now consider the following assignment $\A'$ in $H_A \cup G_B$: $\A'$ is the same as $\A$ on the set $X$ and assigns every vertex in $L \setminus X$ according to the matching $M$. It is easy to check that $\load_{\A'}(s) = \load_{\A}(s)$ for $s \in S_{\A}$ and $\load_{\A'}(s) = 1$ for $s \in R \setminus S_{\A}$, so $\maxload(\A') = \maxload(\A)$, which completes the proof of \autoref{eq:spar-C-intuition} 

Combining Equations \ref{eq:bad-set} and \ref{eq:spar-C-intuition} we have:
\begin{align*}
\optload(H_A \cup G_B) &= \optload(H_A[X \cup R])\\
&> \alpha \cdot \optload(G_A[X \cup R]) = \alpha \cdot \optload(G_A \cup G_B),
\end{align*} 
which contradicts the assumption that $\pi$ is an $\alpha$-approximate one-way communication protocol.

\subsection{Equivalent Notions of Sparsification}\label{sec:equiv-notions-sparsification}

Although our primary motivation for studying load-balancing sparsifiers is to understand the streaming/one-way-communication complexity of load balancing, we believe that they are a very natural combinatorial object with connections to other problems.

Recall that the \textbf{vertex expansion} of a subset $X$ is defined as $\psi(X) := |N(X)|/|X|$ (see e.g.~\cite{HLW06}).
From the characterization of $\optload$ in \autoref{lem:lb-halls},
one can see that load-balancing sparsifiers are closely related to vertex-expansion sparsifiers.

\begin{observation}[Connection with Vertex Expansion] \label{obs:vertex-expansion}
Let $G = (L, R, E)$ be a bipartite graph and let $H = (L, R, E_H)$ be a subgraph of $G$. 
If, 
for every subset $X \subseteq L$,
\[
\psi_H(X) 
= \frac{|N_H(X)|}{|X|} 
\geq \frac{2}{\alpha} \cdot \min\Big\{ \frac{|N_G(X)|}{|X|}, 1 \Big\} 
= \frac{2}{\alpha} \cdot \min\{ \psi_G(X), 1\}, 
\]
then, $H$ is an $\alpha$-load-balancing sparsifier of $G$.
On the other hand, if $H$ is an $\alpha$-load-balancing sparsifier of $G$,
then $\psi_H(X) \geq \frac{1}{2 \alpha} \cdot \min\{ \psi_G(X), 1 \}$ for every subset $X \subseteq L$.
\end{observation}

Note that we cannot replace the right hand side by the simpler expression $\Omega(\frac{1}{\alpha}) \cdot \psi_G(X)$, as otherwise any $\alpha$-sparsifier of the complete bipartite graph $K_{n,n}$ must have at least $\Omega(\frac{n^2}{\alpha})$ edges (because the constraints on the singletons require every vertex on the left to have degree at least $\Omega(\frac{n}{\alpha})$), while one can prove that there is an $\alpha$-load-balancing sparsifier of $K_{n,n}$ with $O(n \log n)$ edges by random sampling.

\begin{remark}[Connection with Matching Conductance]
The $\min\{\psi(X),1\}$ term in \Cref{obs:vertex-expansion} is a natural quantity that is closely related to the notion of \textbf{matching conductance} defined in \cite{Olesker-TaylorZ22}, which was used in analyzing the fastest mixing time of a graph.
Let $\nu(G)$ be the size of a maximum matching in $G$.
The matching conductance of a set $X$ is defined as $\gamma(X) = \nu(E(X,\overline{X}))/|X|$.
For a bipartite graph $G=(L, R, E)$, by Hall's theorem, it can be checked that 
\[
\min_{X \subseteq L} \gamma(X) := \min_{X \subseteq L} \min\{\psi(X),1\}.
\]
\end{remark}

\paragraph{Operational Definition of Sparsifiers:} We will not use vertex expansion and matching conductance in this paper, and so we do not provide details of the connections discussed above. Instead, we will use the following characterization in our proofs, as it provides the easiest way of verifying that a graph $H \subset G$ is indeed a load-balancing sparsifier.

\begin{lemma} \label{lem:equiv-sparsifiers}
Let $G = (L, R, E)$ be a bipartite graph, and $H = (L, R, E_H)$ be a subgraph of $G$. Then, for any $\alpha \geq 1$, the following statements are equivalent.
\begin{enumerate}
	\item\label{item:load} \textit{Load-Balancing:} For every $X \subseteq L$, $\optload(H[X \cup R]) \leq \alpha \cdot \optload(G[X \cup R])$.
	\item\label{item:operational} \textit{Operational Definition:} For every $X \subseteq L$ that is matchable in $G$, $N_H(X) \geq \frac{1}{\alpha} \cdot |X|$.
\end{enumerate}
\end{lemma}

\begin{proof}
It is easy to see that \textbf{(1)} $\rightarrow$ \textbf{(2)}. 
Let $H$ be an $\alpha$-load-balancing sparsifier of $G$.
Consider any matchable set $X \subseteq L$. 
As $X$ is matchable, we have $\optload(G[X \cup R]) = 1$, and so $\optload(H[X \cup R]) \leq \alpha$ as $H$ is an $\alpha$-sparsifier.
It follows from \autoref{lem:lb-halls} that $\card{N_H(X)} \geq \frac{1}{\alpha}\card{X}$.

The other direction is more useful.
For any $X \subseteq L$ with $\optload(G[X \cup R]) = d$, we need to prove that $\optload(H[X \cup R]) \leq \alpha d$. 
By \autoref{lem:lb-halls}, this is equivalent to proving that $\card{N_H(U)} \geq \frac{1}{\alpha d}\card{U}$ for any $U \subseteq X$. 
Let $\A$ be the optimal assignment in $G[U \cup R]$.
Clearly,
\[
\maxload(\A) = \optload(G[U \cup R]) \leq \optload(G[X \cup R]) \leq d.
\]
Let $R_\A \defeq \{s\in R \mid \load_{\A}(s) \geq 1\}$ and note that $\card{R_\A} \geq \card{U}/\maxload(\A) \geq \card{U}/d$. Construct a set of clients $C \subseteq U$ as follows: for each server $s \in R_\A$, add \emph{exactly one} client from $\A^{-1}(s)$ to $C$. 
Then $\card{C} = \card{R_\A} \geq \card{U}/d$ and $C$ is matchable in $G$. 
By the assumed property (\ref{item:operational}) of the lemma, it follows that 
\[
\card{N_H(U)} \geq \card{N_H(C)} \geq \frac{1}{\alpha} \card{C} \geq \frac{1}{\alpha d} \card{U},
\]
concluding the proof. 
\end{proof}

\subsection{\rsplus s} \label{sec:rsplus}

The key conceptual contribution of this paper is showing that the \emph{non-existence} of a load-balancing sparsifier is nearly equivalent to the \emph{existence} of an extremal combinatorial object that we call a \rsplus, defined formally as follows. 

\begin{definition}[\rsplus] \label{dfn:rsplus} 
For any $\alpha \geq 1$, we say a bipartite graph $G = (L, R, E)$ is an \textbf{$\bm{\alpha}$-\rsplus} iff the edge-set $E$ can be partitioned into matchings $M_1\ldots M_k$ such that for each $M_i$, $L(M_i)$ has at most $\card{L(M_i)}/\alpha$ neighbors in $G \setminus M_i$;
i.e., $\card{N_{G \setminus M_i}(L(M_i))} \leq \card{M_i}/\alpha$. 
Note that $M_i$ can be of different sizes and a \rsplus\ contains $\sum_i \card{M_i}$ edges.
\end{definition}

It is not difficult to see that a \rsplus\  cannot be sparsified to preserve the load-balancing properties: if we remove many edges of a matching $M_i$ then $|N_H(L(M_i))| \ll |N_G(L(M_i))|$; see \autoref{lem:easy} for a proof.
So, if there exists a dense \rsplus\ graph, then this provides a lower bound on the size of a load-balancing sparsifier.

\begin{definition} \label{dfn:MC}
Define:
\begin{itemize}
\item $\MC(n,\alpha)$ as the largest possible number of edges in any $\alpha$-\rsplus\ $G = (L, R, E)$ with $\card{L} = n$ (notice that there is no requirement on the size of $R$); 
\item $\MC(G,\alpha)$ for any given graph $G = (L, R, E)$, as the largest number of edges in any $\alpha$-\rsplus\ $H = (L, R, E_H)$ such that $H$ is a subgraph of $G$. 
\end{itemize}
Note that both $\MC(n,\alpha)$ and $\MC(G,\alpha)$ are monotonically \emph{decreasing} as $\alpha$ increases.
\end{definition}

We will prove in \autoref{sec:LP} that $\MC(n,\alpha) \approx \spar(n,\Theta(\alpha))$ up to an $O(\log n)$ factor.
The proof uses randomized rounding and linear programming duality in the same spirit of~\cite{GoelKK12} (for ``matching sparsifiers'' and \rs graphs), although the details use some problem-specific insights. 
This shows in a precise sense that a \rsplus\ is the dual object of a load-balancing sparsifier.

\begin{remark}[\rsplus s and \rs Graphs]
\label{comparison}
\rsplus s are closely related to \rs graphs in the following sense. An $(r,t)$-RS graph is any graph whose edges can be partitioned into $t$ induced matchings of size $r$~\cite{RuzsaS78}. 
We can turn a \rsplus\ $G=(L,R,E)$ with $k$ specified matchings $M_1,\ldots,M_k$ into an RS graphs as follows: for each matching $M_i$ of $G$, remove at most $\card{M_i}/\alpha$ edges to turn it into an induced matching $M'_i$ (remove the edges of $M_i$ incident on vertices of $R$ that are neighbors to $L(M_i)$ 
in $G \setminus M_i$). Then, round down size of each matching to a power of two by removing at most half of its remaining edges. Finally, among these $\Theta(\log{n})$ different classes of matchings (according to their size), pick the one
that contains the largest number of edges overall. This way, we obtain an $(r,t)$-RS graph $G'$ with 
\[
\text{density} = \frac{\card{E}}{\Theta(\log{n})}, \quad t \leq k, \quad \text{and} \quad r \geq \frac{\card{E}}{k \log{n}}.
\] 
However, note that the guarantee of \rsplus\ is a lot stricter which leads to a ``stronger'' type of RS graph $G'$: each matching $M'_i$ is not only induced, which means $L(M'_i)$ avoids $R(M'_i)$ in $G' \setminus M'_i$, but in 
fact $L(M'_i)$ avoids almost the entirety of $R$, except for $\leq \card{M'_i}/\Theta(\alpha)$ vertices. On the other hand, we should note that unlike RS graphs wherein one explicitly fixes size of the matchings to be some parameter $r$ (sometime as 
large as $\Theta(n)$ even), in \rsplus s, there is no explicit bound on the size of the matchings and they can even be of different size. 
\end{remark}


\section{Relating \rsplus s\ to Load-Balancing Sparsifiers} \label{sec:LP}

The key result of this section is a near-equivalence between $\spar(n,\alpha)$ and $\MC(n,\alpha)$. 

\begin{restatable}{theorem}{MCsparthm}	
\label{thm:MC-spar}
For any  integer $n \geq 1$ and $\alpha \geq 2$,  
	\[
	\MC(n,2\alpha) \lesssim \spar(n,\alpha) \lesssim  \MC\big(n,\frac{\alpha}{2}\big) \cdot \ln{(n)}. 
	\]
	
	Moreover, for any bipartite graph $G$, $\spar(G,\alpha) \lesssim \MC(G,\frac{\alpha}{2}) \cdot \ln{(n)}$. 
\end{restatable}

The direction $\MC(n,2\alpha) \lesssim \spar(n,\alpha)$ is easy and gives a good intuition of the definition of a \rsplus.  
It shows that the vertex expansion of a \rsplus\ is very brittle under edge removal.

\begin{lemma}[Easy Direction] \label{lem:easy}
$\MC(n,2\alpha) \lesssim \spar(n,\alpha)$ for any positive integer $n$ and $\alpha \geq 2$.
\end{lemma}
\begin{proof}
Let $G = (L, R, E)$ be an extremal $2\alpha$-\rsplus\  that contains $\MC(n,2\alpha)$ edges. We will show that any $\alpha$-approximate load-balancing sparsifier of $G$ must contain at least half of the edges of $G$; the lemma then follows as

\[
\frac12 \cdot \MC(n,2\alpha) 
= \frac{1}{2} \cdot \card{E(G)}
\leq \spar(G,\alpha) 
\leq \spar(n,\alpha).
\] 

Let $\{M_i\}_{i=1}^t$ be matchings into which $E(G)$ is partitioned according to~\Cref{dfn:rsplus}.
We argue that any $\alpha$-sparsifier $H$ of $G$ must contain at least half of the edges of \emph{every} $M_i$, and this will complete the proof of the lemma.
 
Consider any matching $M_i$ and let $L_i$ be the left endpoints and $R_i$ be the right endpoints. 
Assume, for contradiction, that $\card{M_i \cap E(H)} < \frac12 \card{M_i}$, and let $L'_i \subseteq L_i$ contain the vertices in $L_i$ whose matching edge from $M_i$ is \emph{not} present in $H$.
By our assumption, we have $\card{L'_i} > \card{L_i}/2$. 
Since $G$ is an $2\alpha$-\rsplus, it follows that
\[\card{N_{H}(L'_i)} \leq \card{N_{G \setminus M_i}(L_i)} \leq \frac{\card{L_i}}{2\alpha} < \frac{\card{L'_i}}{\alpha}.
\]
This contradicts with $H$ being an $\alpha$-approximate load-balancing sparsifier of $G$, as $L'_i$ is matchable in $G$ (see Property (\ref{item:operational}) of \autoref{lem:equiv-sparsifiers}).
\end{proof}

The rest of this section is to prove the other direction of \autoref{thm:MC-spar}.
To do so, we will introduce an LP-relaxation of $\spar(G,\alpha)$, and use its primal to relate to $\spar(G,\alpha)$ and use its dual to relate to $\MC(G,\frac{\alpha}{2})$.

\subsection{Linear Programming Relaxation for Load-Balancing Sparsification}

The following is the primal LP that captures the problem of finding an $\alpha$-approximate load-balancing sparsifier of $G$ with the minimum number of edges.

\begin{definition}[Primal LP] \label{def:primal}
The Primal LP is defined for a bipartite graph $G = (L, R, E)$ with $|L| = n$ and a parameter $\alpha \geq 1$. 
There is a variable $x_e$ for every edge $e \in E$. 
Define a pair of sets $X\subseteq L,Y \subseteq R$ to be \emph{contracting} if $X$ is matchable and $\card{Y} \leq \frac{1}{2\alpha} \cdot \card{X}$. 
The LP will have a constraint for every contracting pair $(X,Y)$. 
For input $G$ and $\alpha$, the primal LP is defined as
\begin{align*}
\LP(G,\alpha) ~:=~ \textrm{minimize }  &~~~\sum_{e \in E} x_e  & 
\\
\st &\quad \sum_{e \in E(X,R \setminus Y)} x_e \geq \frac{|X|}{2}  \qquad  \text{for all contracting pair } X,Y  
\\
&\quad x_e \geq 0 \qquad \qquad \qquad \qquad \text{for all edge } e \in E.
\end{align*}
\end{definition}

For intuition about the primal LP,
think of $x_e$ as representing whether edge $e$ is included in the sparsifier $H$. We want $H$ to satisfy Property (\ref{item:operational}) of \autoref{lem:equiv-sparsifiers}, since this is equivalent to $H$ being a load-balancing sparsifier. For any matchable set $X \subseteq L$, this property requires that $X$ should \emph{not} contract to some small set $Y$. In other words, for every small set $Y$, there should be many edges from $X$ to $R \setminus Y$, which is precisely the main constraint of the LP.
 
The following lemma shows that $\LP(G,\alpha)$ is a $O(\log n)$-approximation of $\spar(G,\alpha)$. 
The proof is by a standard randomized rounding argument.

\begin{lemma}[LP and \spar] \label{lem:spar-LP}
For a bipartite graph $G = (L, R,E)$ with $\card{L} = n$ and $\card{R} \leq n^2$, 
\[
\LP(G,\alpha) \leq \spar(G,\alpha) \leq 20\cdot\ln(n)\cdot\LP\big(G,\frac{\alpha}{2}\big).
\]
\end{lemma}

\begin{proof}
We first show that $\LP(G,\alpha) \leq \spar(G,\alpha)$, which says that $\LP(G,\alpha)$ is a relaxation of $\spar(G,\alpha)$. 
Let $H$ be an $\alpha$-sparsifier of $G$ with $\card{E(H)} = \spar(G,\alpha)$. 
We will show that the following is a feasible solution to the LP: $x_e = 1$ if $e \in H$ and $x_e=0$ otherwise. 
Consider any contracting pair of sets $(X,Y)$, i.e.~$X$ is matchable and $|Y| \leq \frac{1}{2\alpha} |X|$.
Define 
\[
X' := \set{v \in X \mid N_H(v) \subseteq Y}.
\]
We must have $\card{X'} \leq \frac12 |X|$, otherwise $|Y| < \frac{1}{\alpha} |X'|$ and $X'$ would violate Property (\ref{item:operational}) of~\Cref{lem:equiv-sparsifiers} for an $\alpha$-approximate load-balancing sparsifier. 
Thus, in $H$, at least $\card{X} - \card{X'} \geq \card{X}/2$ vertices of $X$ have a neighbor in $R \setminus Y$; so, the LP constraint for the contracting pair $(X,Y)$ is satisfied.

We next prove that $\spar(G,\alpha) \leq 20 \cdot \ln(n) \cdot \LP(G,\frac{\alpha}{2})$, by showing that a fractional solution to $\LP(G,\frac{\alpha}{2})$ can be rounded to an integral solution to $\spar(G,\alpha)$ with at most $20 \cdot \ln(n) \cdot \LP(G,\frac{\alpha}{2})$ edges.
Given a feasible solution $x_e$ to the $\LP(G,\frac{\alpha}{2})$, we construct an $\alpha$-approximate load-balancing sparsifier $H$ by a simple randomized rounding procedure as follows. 
Define $p = 10\ln(n)$. 
For every $e \in G$, add $e$ to $H$ with probability $\min\{p \cdot x_e, 1\}$. 
Note that $H$ has at most $p \sum_{e} x_e = p \cdot \LP(G,\frac{\alpha}{2})$ edges in expectation. 
So, by Markov's inequality, 
\[
\Pr\Big[\card{E(H)} \leq 20 \cdot \ln(n) \cdot \LP\big(G,\frac{\alpha}{2} \big)\Big] \geq \frac12.
\] 
We complete the proof by arguing that, with high probability, $H$ satisfies Property (\ref{item:operational}) of \autoref{lem:equiv-sparsifiers},
which would imply that an $\alpha$-sparsifier with the desired number of edges exists by the probabilistic method. 
We need to show that, with high probability, for any matchable $X \subseteq L$ and $Y \subseteq R$ with $\card{Y} < \frac{1}{\alpha} \cdot \card{X}$, we have $E_H(X,R \setminus Y) \neq \emptyset$. 
Consider any such sets $X,Y$ and denote $k \defeq \card{X}$. 
As we are considering $\LP(G,\frac{\alpha}{2})$, the pair $(X,Y)$ is a contracting pair.
So, by the LP constraint, we have $\sum_{e \in E(X,R \setminus Y)} x_e \geq {k}/{2}$.
If any edge $e$ in $E(X,R \setminus Y)$ has $p \cdot x_e \geq 1$, then this edge $e$ will be added to $H$ with probability $1$ and we are done.
Henceforth we assume that $px_e < 1$ for all $e \in E(X, R \setminus Y)$. 
Since edge $e$ is sampled independently with probability $p x_e$, it follows that
\[
\Pr\big[E_H(X,R \setminus Y) = \emptyset\big] 
= \prod_{e \in E(X, R\setminus Y)} (1-px_e) 
\leq \exp\bigg(-\sum_{e \in E(X, R \setminus Y)} px_e\bigg) 
\leq \exp\Big(-\frac{pk}{2}\Big) 
= n^{-5k}.
\] 
Fixing any particular $k = \card{X}$, the number of such set-pairs $(X,Y)$ is clearly at most $n^k \cdot |R|^k \leq n^{3k}$, using the assumption of the lemma that $\card{R} \leq n^2$. 
Therefore, by a union bound, we have $E_H(X,R \setminus Y) \neq \emptyset$ for $\card{X} =k$ with probability at least $1-n^{-2k} \geq 1-n^{-2}$. 
Applying another union bound over all possible $k$ yields probability at least $1-\frac1n$ for all such set-pairs $(X,Y)$.
\end{proof} 

Our next goal is to relate $\LP(G,\alpha)$ to $\MC(G,\alpha)$ (recall \autoref{dfn:MC}), for which we consider the dual of the above linear program.

\begin{definition}[Dual LP] \label{def:dual}
The dual LP is again defined for a bipartite graph $G = (L, R, E)$ with $|L| = n$ and a parameter $\alpha \geq 1$. 
There is a variable $y_{X,Y}$ for every contracting set-pair $(X,Y)$, 
and a constraint for every edge $e \in E$.
\begin{align*}
\LP(G,\alpha) ~:=~ \textrm{minimize }  &~~~\frac{1}{2}\sum_{\textrm{contracting~} (X,Y)} |X| \cdot y_{X,Y}   & 
\\
\st &\quad \sum_{\substack{\textrm{contracting~} (X,Y), \\ e \in E_G(X,R\setminus Y) }} y_{X,Y} \leq 1 \qquad \textrm{ for all } e \in E  \\
&\quad y_{X,Y} \geq 0 \qquad \qquad \qquad \qquad ~~~\text{for all contracting pair } (X,Y)
\end{align*}
\end{definition}

It is straightforward to check that the LP in \autoref{def:dual} is indeed the dual program of the LP in \autoref{def:primal}, and so by the strong LP duality theorem they have the same objective value.
The following lemma shows that one can construct a \rsplus\ from the dual LP solution.

\begin{lemma}[LP and \rsplus] \label{lem:LP-MC}
For any $G = (L, R, E)$ with $\card{L} = n$ and $\alpha \geq 2$, 
\[
\LP(G,\alpha) \leq 20 \cdot \MC(G,\alpha).
\]
\end{lemma}

The proof is again by a randomized rounding procedure, but the details are more involved, and we dedicate the next subsection to it.

We end this subsection by showing that the hard direction of \Cref{thm:MC-spar} follows immediately from \Cref{lem:spar-LP} and \Cref{lem:LP-MC}.

\begin{proofof}{\Cref{thm:MC-spar} assuming~\Cref{lem:LP-MC}}
One direction is already proved in \autoref{lem:easy}.
For the other direction, we first apply \autoref{claim:truncate-servers} to reduce the problem to a bipartite graph $G = (L, R, E)$ with $|L|=n$ and $|R| \leq n^2$; this claim is quite trivial, so we defer the formal statement and proof to \Cref{sec:reducing-num-servers} in the appendix.
Then, \Cref{lem:spar-LP} and \Cref{lem:LP-MC} imply that
\[
\spar(G,\alpha) 
\leq 20 \cdot \ln(n) \cdot \LP\big(G,\frac{\alpha}{2}\big) 
\leq 400 \cdot \ln(n) \cdot \MC(G,\frac{\alpha}{2}),
\]
which also implies that $\spar(n,\alpha) \leq 400\cdot\ln(n)\cdot\MC(n,\alpha/2)$.
\end{proofof}

\subsection{Constructing \rsplus\  from Dual Solution}

The goal of this subsection is to prove \autoref{lem:LP-MC}.
We first provide some intuitions of the proof in \autoref{sec:ideas} by considering the ideal case when all dual variables $y_{X,Y}$ are integral. We then present the general construction in \autoref{sec:construction-from-fractional}
 and the analysis in \autoref{sec:analysis}.

\subsubsection{Proof Ideas} \label{sec:ideas}

We start with a simple definition and a simple observation.

\begin{definition} \label{dfn:deviates}
Given a pair of sets $(X,Y)$ with $X \subseteq L$ and $Y \subseteq R$, we say that an edge $(u,v)$ \textbf{deviates from} $(X,Y)$ if $u \in X$ and $v \in R \setminus Y$.
\end{definition}

To illustrate some proof ideas, 
we consider the ideal case when all dual variables $y_{X,Y}$ are integral. 
By the following observation, each $y_{X,Y}$ is either $0$ or $1$.

\begin{observation}
	\label{obs:less-than-one}
Every dual variable $y_{X,Y}$ is at most $1$.
\end{observation}

\begin{proof}
Suppose for contradiction that $y_{X,Y} > 1$. By the definition of contracting pair $(X,Y)$, $X$ is matchable by some matching $M$ and $X$ is larger than $Y$.
Thus, there must be some edge $(u,v) \in M$ that deviates from $(X,Y)$, 
but then the dual constraint of $(u,v)$ is violated.
\end{proof}

\paragraph{Intuition: Construction from Integral Dual Solution.} 
We create a \rsplus\ from a feasible dual $\{0,1\}$-solution as follows.
Let $\mathP$ contain all contracting pairs $(X,Y)$ for which $y_{X,Y} = 1$. 
For every $(X,Y) \in \mathP$, let $M_{X,Y}$ be some matching from $X$ to $R$, 
which exists by the definition of contracting pairs (see \autoref{def:primal})\footnote{To avoid confusion, note that $Y$ does not correspond to the right endpoints of $M_{X,Y}$; instead, $Y$ relates the matching to the dual variable $y_{X,Y}$.}.
Then, we remove from $M_{X,Y}$ all edges that are incident to $Y$, 
and let $M'_{X,Y}$ denote the remaining matching. 
Let $H$ be the union of all the $M'_{X,Y}$.
Clearly, $H$ is a subgraph of $G$, and that is the complete construction in this simpler setting.

We claim that $H$ is an $\alpha$-\rsplus\ with $\gtrsim \LP(G,\alpha)$ edges.
First, we lower bound the number of edges in $H$.
Since $\alpha \geq 2$ by the assumption of \autoref{lem:LP-MC} and $Y \leq \frac{1}{2\alpha}\card{X}$ by the definition of a contracting pair, it follows that $\big| M'_{X,Y} \big| = |X| - |Y| \geq \frac34 \card{X}$.
In the full proof in \autoref{sec:analysis}, we will show that the matchings $M'_{X,Y}$ are all edge-disjoint.
Therefore,
\[
\LP(G,\alpha) 
= \frac{1}{2}\sum_{(X,Y) \in \mathP} \card{X} 
\leq \frac{2}{3} \sum_{(X,Y) \in \mathP} \card{M'_{X,Y}} 
= \frac{2}{3} \cdot |E(H)|,
\]
where the first equality is by the observation that $y_{X,Y} \in \{0,1\}$,
and the last equality is by the fact that the matchings $M'_{X,Y}$ are edge-disjoint.

It remains to argue that $H$ satisfies the properties of an $\alpha$-\rsplus\ in \autoref{dfn:rsplus}. 
Consider some matching $M := M'_{X,Y}$ in $H$.
Let $X' \subseteq X$ be the left endpoints $L(M)$. 
Recall that $\card{X'} \geq \frac{3}{4} \card{X}$. 
Suppose for contradiction that $X'$ has more than $\frac{1}{\alpha} \card{X'}$ neighbors in $H \setminus M$. 
Then 
\[
\card{N_{H \setminus M}(X')} > \frac{1}{\alpha} \card{X'} \geq \frac{3}{4\alpha}\card{X}.
\]
Since $\card{Y} \leq \frac{1}{2\alpha}\card{X}$, 
there must be some edge $(v,z)$ in $H \setminus M$ such that $v \in X' \subseteq X$ and $z \notin Y$,
and so $(v,z)$ deviates from $(X,Y)$. 
The edge $(v,z)$ must come from some other matching $M'_{P,Q}$ of $H$. 
By our construction of matchings, we must have $y_{P,Q} =1$,
and also none of the edges of $M'_{P,Q}$ are incident to $Q$.
This implies that the edge $(v,z)$ also deviates from $(P,Q)$ with $v \in P$ and $z \notin Q$. 
So, the edge $(v,z)$ deviates from both $(X,Y)$ and $(P,Q)$ with $y_{X,Y} = y_{P,Q} = 1$, but this means that the dual constraint for $(v,z)$ is violated, arriving at our contradiction.
We conclude that $H$ is an $\alpha$-\rsplus.

\subsubsection{Construction from Fractional Dual Solution} \label{sec:construction-from-fractional}

In general, we have a \emph{fractional} solution to the Dual LP of value $\LP(G,\alpha)$. 
We will use randomized rounding to construct a subgraph $H$ of $G$ and argue that $H$ is the desired \rsplus\ with positive probability. 
We shall note that unlike the argument for the Primal LP in~\Cref{lem:spar-LP}, here, the Dual LP has a very large integrality gap and thus the rounding
should be \emph{bicriteria} (this will become more clear shortly).

We sample every contracting pair $(X,Y)$ with probability $\frac{1}{10} \cdot y_{X,Y}$ and let $\mathP$ be the set of all sampled contracting pairs. 
For every $(X,Y) \in \mathP$, let $M_{X,Y}$ be a perfect matching from $X$ to $R$ in $G$; such a matching must exist by the definition of contracting pairs in \autoref{def:primal}. 
We will fix in advance a matching for every matchable set $X$, so that the choice of edges in matching $M_{X,Y}$ does not depend on the dual variables or any of our random choices.
As before, for every $M_{X,Y}$, remove all edges in $M_{X,Y}$ that are incident to $Y$, and let $M'_{X,Y}$ be the remaining matching. 
In the ideal case above,  
the union of these $M'_{X,Y}$ is a \rsplus.
In the general case, however, this is not necessarily true
and we will do the following post-processing step to obtain a \rsplus.

\begin{definition} \label{def:overloaded}
We say that an edge $(u,v) \in G$ is \textbf{overloaded} if there exist two different set-pairs $(X,Y) \in \mathP$ and $(P,Q) \in \mathP$ such that $(u,v)$ deviates from both $(P,Q)$ and $(X,Y)$.  
\end{definition}

In the postprocessing step,
for every $M'_{X,Y}$, we remove all edges in $M'_{X,Y}$ that are overloaded, and let $M''_{X,Y} \subseteq M'_{X,Y}$ be the resulting matching.
We say $M''_{X,Y}$ is \textbf{good} if $\big| M''_{X,Y} \big| \geq \frac12 \card{X}$.
Our final graph $H$ will consist of the union of all good $M''_{X,Y}$. 
The construction is summarized in~\Cref{alg:contractor}.

\begin{algorithm}[ht]\caption{Construction of \rsplus\ from Dual LP Solution}
 \label{alg:contractor}
  {\bf Input:} a solution $\{y_{X,Y}\}_{\textrm{contracting}~(X,Y)}$ to the dual LP with objective value $\LP(G,\alpha)$. \\
\begin{enumerate}
	\item Sample every contracting pair $(X,Y)$ into $\mathP$ with probability $\frac{1}{10} \cdot y_{X,Y}$.
	\item {\bf For} every $(X,Y) \in \mathP$ do:
	\begin{enumerate}
		\item Let $M_{X,Y}$ be an arbitrary matching from $X$ to $R$.
		\item\label{item:prime} Construct $M'_{X,Y} \subseteq M_{X,Y}$ by removing from $M_{X,Y}$ all edges that are incident to $Y$.
		\item\label{item:doubleprime} Construct $M''_{X,Y} \subseteq M'_{X,Y}$ by removing all overloaded edges from $M'_{X,Y}$ as defined in \autoref{def:overloaded}. 
Label $M''_{X,Y}$ as \emph{good} if $\big| M''_{X,Y} \big| \geq \frac12 \card{X}$.
	\end{enumerate}
\end{enumerate}
{\bf Output:} the graph $H$ that is the union of all the \emph{good} matchings $M''_{X,Y}$.
\end{algorithm}

\subsubsection{Analysis} \label{sec:analysis}

As in \autoref{sec:ideas}, we will lower bound the number of edges in $H$, and prove that $H$ is an $\alpha$-\rsplus.
The following claim will help us upper bound the number of edges that we remove in the post-processing step.

\begin{claim} \label{claim:overloaded}
For an edge $(u,v)$ in a matching $M'_{X,Y}$, 
\[
\Pr\big[(u,v) \textrm{ is overloaded~}\big] \leq \frac{1}{10}.
\]
\end{claim}

\begin{proof}
Consider all contracting pairs $(P,Q) \neq (X,Y)$ such that $(u,v)$ deviates from $(P,Q)$. 
The edge $(u,v)$ can be overloaded only if one of these $(P,Q)$ is also sampled into $\mathcal{P}$. 
Since each $(P,Q)$ is sampled independently with probability $\frac{1}{10}y_{P,Q}$ (and in particular independent from $(X,Y)$), 
it follows from the union bound and the dual constraint that
\[
\Pr\big[(u,v) \textrm{ is overloaded~}\big] 
\leq \sum_{(P,Q) \neq (X,Y)~\mid~(u,v) \textrm{ deviates from } (P,Q)} \frac{1}{10} \cdot y_{P,Q} \leq \frac{1}{10}. \qed
\]

\end{proof}	

By the same argument as in \autoref{sec:ideas}, each $\big|M'_{X,Y} \big| \geq \frac34 \card{X}$.
So, it follows from \autoref{claim:overloaded} and Markov's inequality that many $M''_{X,Y}$ are good.

\begin{observation} \label{obs:probably-good}
Every $M''_{X,Y}$ is good with probability at least $\frac12$. 
\end{observation}

This allows us to lower bound the number of edges in $H$.

\begin{claim} \label{claim:LP-num-edges}
With positive probability, 
\[
\card{E(H)} \geq \frac{1}{20} \cdot LP(G,\alpha).
\]
\end{claim}

\begin{proof}
First we show that the matchings $M''_{X,Y}$ in $H$ are edge-disjoint. 
Suppose for contradiction that some $(u,v)$ is in both $M''_{X,Y}$ and $M''_{P,Q}$, where $(X,Y), (P,Q) \in \mathP$. 
By Step (\ref{item:prime}) of~\Cref{alg:contractor}, 
it follows that $(u,v)$ deviates from both $(X,Y)$ and $(P,Q)$.
But this means that $(u,v)$ is overloaded, which contradicts the removal of overloaded edges in Step (\ref{item:doubleprime}).

Consider any contracting pair $(X,Y)$.
If the pair is sampled into $\mathP$, 
then the matching $M''_{X,Y}$ is good with probability at least $\frac12$ by \autoref{obs:probably-good}, in which case $\big|M''_{X,Y}\big| \geq \frac12 \card{X}$ edges are added to $H$. 
Each contracting pair $(X,Y)$ is sampled into $\mathP$ with probability $\frac{1}{10} \cdot y_{X,Y}$. 
Thus, in expectation, each contracting pair $(X,Y)$ contributes $\frac{1}{40} \cdot y_{X,Y} \cdot |X|$ to $E(H)$. 
As the matchings $M''_{X,Y}$ are edge-disjoint,
\[
\Exp\card{E(H)}
\geq \sum_{\textrm{contracting } (X,Y)} \frac{1}{40} \cdot y_{X,Y} \cdot |X|
= \frac{1}{20}\cdot\LP(G,\alpha).
\] 
We conclude that there exists such an $H$ that satisfies the statement in the claim.
\end{proof}

We finish the proof of \Cref{lem:LP-MC} by showing that $H$ is an $\alpha$-\rsplus.
The argument is similar to that in \Cref{sec:ideas}.

\begin{claim}
\label{claim:LP-rsplus}
$H$ satisfies the properties of an $\alpha$-\rsplus\ in \Cref{dfn:rsplus}.
\end{claim}

\begin{proof}
Consider some matching $M := M''_{X,Y}$ and let $X'' \subseteq X$ be the left endpoints $L(M)$. 
Since $M$ is good by construction, we have $\card{X''} \geq \frac12 \card{X}$. 
Suppose for contradiction that $X''$ has more than $\frac{1}{\alpha} \card{X''}$ neighbors in $H \setminus M$. 
Then $\card{N_{H \setminus M}(X'')} > \frac{1}{\alpha} \card{X''} \geq \frac{1}{2 \alpha} \card{X}$.
Since $\card{Y} \leq \frac{1}{2\alpha}\card{X}$ by the definition of a contracting pair, there must be some edge $(v,z)$ in $H \setminus M$ such that $v \in X' \subseteq X$ and $z \notin Y$, and so $(v,z)$ deviates from $(X,Y)$. 
Since $(v,z) \in H \setminus M$, it must come from some other matching $M''_{P,Q}$ of $H$.
This implies that $(P,Q) \in \mathP$. 
Moreover, $(v,z)$ deviates from $(P,Q)$, by Step (\ref{item:prime}) of the construction algorithm. 
Thus, the edge $(v,z)$ deviates from two different pairs in $\mathP$ -- namely $(X,Y)$ and $(P,Q)$ -- which means that $(v,z) \in H$ is overloaded, but this contradicts with the removal of overloaded edges in Step (\ref{item:doubleprime}) of the construction algorithm.
\end{proof}

\autoref{lem:LP-MC} follows immediately from 
\autoref{claim:LP-rsplus} and \autoref{claim:LP-num-edges}.


\section{A Construction of \rsplus s}\label{sec:construction}

In this section, we present a simple construction of somewhat dense \rsplus s.
\begin{theorem}[Density of \rsplus s] \label{thm:rsplus}
For $\eps \in (0,1)$ a sufficiently small constant and $n \geq 1$ a sufficiently large integer,
there are $(n^{\frac{1}{4} - O(\eps)})$-\rsplus s $G=(L, R, E)$ with $n^{1+\Omega(\eps^2)}$ edges where $|L| =: n$ and $|R| = \sqrt{n}$.
\end{theorem}

The bipartite graph $G=(L, R, E)$ that we construct has $|L|=w^{2k}$ vertices on the left and $|R|=w^k$ vertices on the right, for some integers $w, k \geq 2$.
We think of each vertex in $L$ as a string of length $2k$ where each character is in $\{1,\ldots,w\}$, and similarly each vertex in $R$ as a string of length $k$ where each character is in $\{1,\ldots,w\}$.
The general idea is to use the ``block structure'' of the vertices to argue about the contraction property of the matchings added as we will see.

The edge set in $G$ is simple to describe.
We add edges to $G$ in $t$ rounds.
In each round, we choose a subset of indices $S \subseteq [2k]$ with $|S|=k$.
For each vertex $v \in [w]^{2k}$ in $L$,
we define $v_S$ to be the subsequence of $v$ of length $k$ by restricting $v$ to the indices in $S$ (e.g.~if $v = \{2,5,3,8\}$ and $S = \{1,4\}$ then $v_S = \{2,8\}$).
Note that $v_S \in [w]^k$ corresponds to the unique vertex in $R$ with the same string in $[w]^{k}$.
In each round, for each vertex $v \in L$, we add the edge $(v,v_S)$ in the graph $G$ where $v_S \in R$.
To establish the contraction property, we will choose subsets $S_1, S_2, \ldots, S_t$ where the pairwise intersection size $|S_i \cap S_j|$ is small for $i \neq j$, and run the above process for $t$ rounds.

\begin{algorithm}[ht]\caption{Construction of \rsplus s}
 \label{alg:rsplus}
  {\bf Input:} an integer $w \geq 2$, an integer $k \geq 2$, and $t$ subsets $S_1, \ldots, S_t \subseteq [2k]$ where $|S_i| = k$ for $1 \leq i \leq t$ and $|S_i \cap S_j| \leq \ell$ for $1 \leq i \neq j \leq t$.\\
  {\bf Initialization:} $L = [w]^{2k}$, $R = [w]^k$, and $E = \emptyset$.\\
  {\bf For} $i$ from $1$ to $t$ do:\\
  \quad \quad
  for each vertex $v \in [w]^{2k}$ in $L$, add the edge $(v,v_{S_i})$ to $E$ where $v_{S_i} \in [w]^k$ is in $R$.\\
  {\bf Output:} the graph $G = (L, R, E)$.
\end{algorithm}

To see that $G$ is a \rsplus,
we will partition the edges added in each round into $w^k$ matchings of size $w^k$ as follows,
where each matching connects a disjoint subset of $L$ to $R$.
In each round, when we fix a subset $S \subseteq [2k]$ of size $k$,
we also consider the complement $\overline{S} := [2k]\setminus S$ and the subsequence $v_{\overline{S}}$ restricting a string $v \in [w]^{2k}$ to the subset $\overline{S}$.
For each string $x \in [w]^{k}$ of length $k$,
we define 
\[
L_{S,x} := \big\{v \in [w]^{2k} \mid v_{\overline{S}} = x\big\}.
\]
In words, each group $L_{S,x}$ is the subset of vertices in $L$ where we fix the subsequence in $\overline{S}$ to be $x$.
Then $\{L_{S,x}\}_{x \in [w]^k}$ is a partition of $L$ into $w^k$ groups, with one group for each possible $x$, and each group has $w^k$ vertices.
Note that in \Cref{alg:rsplus}, for each group $L_{S,x}$, we added a perfect matching $M_{S,x}$ from $L_{S,x}$ to $R$ where $M_{S,x} := \{ (v,v_S) \mid v \in L_{S,x} \}$ is of size $w^k$.
The edge set added in each round is the union of $M_{S,x}$ over all $x \in [w]^k$,
and the edge set in the output is the union of $t \cdot w^k$ matchings such that
\begin{equation} \label{e:edge-set}
E = \bigcup_{i: 1 \leq i \leq t} ~\bigcup_{x: x \in [w]^k} M_{S_i,x}.
\end{equation}

The reason that we choose $S_1, \ldots, S_t$ to have pairwise intersection size $|S_i \cap S_j| \leq \ell$ for $i \neq j$ is explained in the following lemma.

\begin{lemma}[Contraction Property] \label{lem:contraction}
In the output graph $G = (L, R, E)$ of \Cref{alg:rsplus},
for each $S_i$ and each $x \in [w]^k$,
the neighbor set of $L_{S_i,x}$ in $G \setminus M_{S_i,x}$ has size
\[
\big|N_{G \setminus M_{S_i,x}}(L_{S_i,x})\big| 
\leq \sum_{j: 1 \leq j \neq i \leq t} w^{|S_j \cap S_i|} \leq t \cdot w^\ell.
\]
\end{lemma}

\begin{proof}
Fix $S_i$ and $x \in [w]^k$.
In round $j \neq i$, 
for each vertex $v \in L$,
we add the edge $(v,v_{S_j})$ to $E$,
where the neighbor of $v$ depend only on the values of $v$ in the indices in $S_j$.
By the definition of $L_{S_i,x}$, 
every vertex $v$ in $L_{S_i,x}$ has the same values in $v_{\overline{S_i}}$ such that $v_{\overline{S_i}} = x$.
In particular, every vertex $v$ in $L_{S_i,x}$ has the same values in $v_{ \overline{S_i} \cap S_j}$, and thus the neighbors of $L_{S_i,x}$ in round $j$ is contained in the set $\{ v_{S_j} \in [w]^k \mid v_{S_j \cap \overline{S_i} } \textrm{~is~fixed} \}$.
This set has size exactly $w^{|S_i \cap S_j|}$ since there are $w$ possible choices for each index in $S_j \cap S_i$ but only one choice for each index in $S_j \cap \overline{S_i}$.
Summing over all $j \neq i$ gives the first inequality of the lemma,
and the second inequality follows by the assumption that $|S_i \cap S_j| \leq \ell$ for $j \neq i$.
\end{proof}

To instantiate our construction, we need a large set family with small pairwise intersection. The existence of such a family 
is a standard result in extremal set theory dating back to the work of Erd\H{o}s and R{\'e}nyi in~\cite{ErdosR56} (e.g., it follows from Gilbert-Varshamov bound in coding theory). 
For completeness, we present a short proof in~\Cref{app:prop-set-family} using a standard probabilistic argument.

\begin{proposition}[Set Family with Small Pairwise Intersection] \label{prop:set-family}
Let $0 < \delta < \frac12$ be any constant and let $c_\delta := 2\cdot \delta^{\delta}\cdot(1-\delta)^{(1-\delta)}$.  
There exists a set family $\mathcal{F} \subseteq {[2k]\choose k}$ of size $\Theta(k^{-\frac14} \cdot c_\delta^k)$ such that for all $S \neq S'\in \mathcal{F}$, it holds that $|S\cap S'|\leq (1-\delta) k$.
\end{proposition}

For two random $k$-subsets of $[2k]$, the expected size of intersection is $\frac{k}{2}$, and so we cannot hope to have $\delta > \frac12$ in the above statement.
Moreover, $c_\delta = 1$ at $\delta = \frac12$ and is monotonically increasing towards $2$ as $\delta\rightarrow 0$. 
As long as $\delta$ is bounded away from $1/2$, size of $|\mathcal{F}|$ is exponential in $k$.

We can now conclude the proof of~\Cref{thm:rsplus} by combining \autoref{lem:contraction} and \autoref{prop:set-family}.

\begin{proofof}{\Cref{thm:rsplus}}
Set $\delta := \frac12 - 2\eps$.
By \autoref{prop:set-family}, there is a set family $\mathcal{F}$ with $\Theta(k^{-\frac14} \cdot c_\delta^k)$ subsets of size $k$ of $[2k]$ with pairwise intersection size at most $(1-\delta)k = (\frac{1}{2} + 2\eps)k$.
Use this set family $\mathcal{F}$ as an input to \autoref{alg:rsplus} so that $t = \Theta(k^{-\frac14} \cdot c_\delta^k)$ and $\ell = (\frac12 + 2\eps)k$.
We will show below that $c_\delta \asymp e^{8\eps^2}$. \footnote{The notation $f(n) \asymp g(n)$ means that $f(n)$ and $g(n)$ are asymptotically equal.}

The edge set of $G$ is the union of $M_{S_i,x}$ for $S_i \in \mathcal{F}$ and $x \in [w]^k$ as shown in~\Cref{e:edge-set}, where $n := |L| = w^{2k}$ in our construction.
By definition, $L_{S_i,x} = L(M_{S_i,x})$.
By setting $w=2$ and $k = \frac12 \log_2 n$, it follows from \autoref{lem:contraction} that
\[
\frac{\big|N_{G \setminus M_{S_i,x}}(L(M_{S_i,x}))\big|}{|M_{S_i,x}|}
\leq \frac{t \cdot w^l}{w^k}
\lesssim \frac{k^{-\frac14} \cdot c_\delta^k \cdot w^{(\frac12+2\eps)k}}{w^k}
\leq \frac{e^{8\eps^2 k}}{w^{(\frac12 - 2\eps)k}}
= \frac{n^{\Theta(\eps^2)}}{n^{\frac14 - \eps}}
= \frac{1}{n^{\frac14 - O(\eps)}},
\]
and so the bipartite graph $G$ is a $(n^{\frac14 - O(\eps) })$-\rsplus.

The number of edges in $G$ is $t \cdot n \asymp k^{-\frac14} \cdot c_\delta^k \cdot n \asymp \log^{-\frac14} n \cdot n^{\Theta(\eps^2)} \cdot n = n^{1+\Theta(\eps^2)}$, by setting $w=2$ and $k= \frac12 \log_2 n$ and using that $n$ is sufficiently large.

It remains to show that $c_{\delta} \asymp e^{8 \eps^2}$ for $\eps$ a sufficiently small constant.
\begin{align*}
c_{\delta} 
& = c_{\frac12 - 2\eps}
= 2\Big(\frac12-2\eps\Big)^{\frac12-2\eps} \Big(\frac12+2\eps\Big)^{\frac12+2\eps} 
= 2\sqrt{\Big(\frac12-2\eps\Big)\Big(\frac12+2\eps\Big)} \cdot \Big(\frac{\frac12+2\eps}{\frac12-2\eps}\Big)^{2\eps}\\
& = 2\sqrt{\frac14 - 4\eps^2} \cdot \Big(\frac{1+4\eps}{1-4\eps}\Big)^{2\eps}
= \sqrt{1-16\eps^2} \cdot \Big(\frac{1+4\eps}{1-4\eps}\Big)^{2\eps}
\asymp  e^{-8\eps^2}\cdot e^{8\eps \cdot 2\eps}
= e^{8\eps^2},
\end{align*}
where we used $e^x \asymp 1+x$ when $x$ is sufficiently small.
This completes the proof.
\end{proofof}


\section{Communication Complexity of Load-Balancing}\label{sec:cc-lb}

We now prove our lower bound on the randomized communication complexity of the load-balancing problem using
the construction of \rsplus s developed in~\Cref{sec:construction}. 

\begin{theorem}\label{thm:cc-lb}
	For any $n \geq 1$ and approximation ratio $\alpha \geq 1$, any two-player one-way randomized communication protocol for finding an $\alpha$-approximation to load-balancing
	with probability of success at least $\frac23$ requires $\Omega(\frac{1}{\log n} \cdot \MC(n,\alpha))$ bits of communication. That is, 
	\[
		\randC{\LB(n,\alpha)} \gtrsim \frac{1}{\log{n}} \cdot \MC(n,4\alpha). \footnote{The constant $4$ in the $\MC(n,4\alpha)$ term in this theorem  
can be replaced with any other constant strictly larger than two. Since the choice of the constant is immaterial for our purpose, we have not attempted to optimize it.} 
	\]
\end{theorem}

Combining this with our construction of \rsplus s in~\Cref{thm:rsplus} on one hand, and the standard reduction from communication to streaming lower bounds on the other hand, gives the following corollaries. 

\begin{corollary}\label{cor:cc-lb}
	For any sufficiently large $n \geq 1$ and sufficiently small $\eps > 0$, 
	\[
		\randC{\LB(n,n^{\frac14-O(\eps)})} = n^{1+\Omega(\eps^2)}.
	\]
	In particular, obtaining any $n^{\frac14-o(1)}$-approximation to load-balancing requires strictly more than any $O(n \cdot \poly\!\log{(n)})$ communication. 
\end{corollary}

\begin{corollary}\label{cor:stream-lb}
	There is no semi-streaming algorithm for obtaining a $n^{\frac14-o(1)}$-approximation to the load-balancing problem with probability of success at least $\frac23$. 
\end{corollary}

The rest of this section is dedicated primarily to the proof of~\Cref{thm:cc-lb}. 
We start by defining a simple family of graphs
that ``encode'' different strings inside \rsplus s. 
We then use these graphs to define our hard input distribution.
After that, we recall some basic information theory and use them to derive the proof of~\Cref{thm:cc-lb}.
We provide short and standard proofs of~\Cref{cor:cc-lb,cor:stream-lb}. Finally, we conclude with 
the complete proof of~\Cref{thm:equivalence} that shows the equivalence between load balancing sparsifiers and one-way communication complexity of load-balancing in~\Cref{sec:equivalence-proof}.

\subsection{Encoding Graphs and the Hard Input Distribution}

 Let $G_0 := (L_0, R_0, E_0)$ be a $(4\alpha)$-\rsplus with 
\[
	\card{L_0} = n \qquad \text{and} \qquad m_0 \gtrsim \frac{1}{\log{n}} \cdot \MC(n,4\alpha)
\]
edges and matchings $M^0_1,\ldots,M^0_k$ such that there exists an integer $r_0 \in [n]$ with 
\[
r_0 \leq \card{M^0_i} < \frac43 \cdot r_0 \quad \text{for all $i \in [k]$}.
\]
The existence of such a graph 
follows by grouping matchings of any $\alpha$-\rsplus\ with density $\MC(n,4\alpha)$ based on sizes of matchings relative to powers of $(4/3)$, and picking the group
with the largest number of edges. 

Our {encoding graphs} are defined as follows. 

\begin{definition}[Encoding Graphs]\label{def:encoding}
Fix a graph $G_0 := (L_0, R_0,E_0)$ as described above. 
Let $x \in \set{0,1}^{E_0}$ be any string whose entries are indexed by edges in $E_0$. Define the \textbf{encoding graph} $G_x := (L,R,E_x)$ of $x$ inside $G_0$ as follows: 
\begin{itemize}
	\item $L := L_0$ -- we use vertices $u \in L_0$ and $u \in L$ interchangeably. 
	\item $R := R_0 \times \set{0,1}$ -- vertices in $R$ are denoted by $v_0$ and $v_1$ for $v \in R_0$; 
	\item $E_x$: for every edge $e=(u,v) \in E_0$, there is exactly one of the edges $(u,v_{0})$ or $(u,v_1)$ depending on whether $x_e = 0$ or $x_e = 1$, respectively. 
\end{itemize}
\end{definition}

\begin{observation}\label{obs:M0-M1}
For any $x \in \set{0,1}^{E_0}$, the graph $G_x$ is a $(2\alpha)$-\rsplus\ with matchings $M_1,\ldots,M_k$, where each $M_i$ is obtained from $M^0_i$ by mapping the edge $(u,v) \in M^0_i$ to either $(u,v_0)$ or $(u,v_1)$, 
depending on whichever one exist in $G_x$. 
\end{observation}
\begin{proof}
	Consider a choice of $G_x$ and one of its designated matchings $M_i$ for $i \in [k]$. By construction,  
	\[
		N_{G_x \setminus M_i}(L(M_i)) \subseteq N_{G_0 \setminus M^0_i}(L(M^0_i)) \times \set{0,1}. 
	\]
	The proof follows since $|N_{G_0 \setminus M^0_i}(L(M^0_i))| \leq \frac{1}{4\alpha} \card{M^0_i}$ as $G^0$ is a $(4\alpha)$-\rsplus. 
\end{proof}

Our hard distribution of inputs is defined as follows. We emphasize that the graph $G_0$ is fixed throughout and both players know this graph. 

\begin{ourbox}
\paragraph{Input distribution $\mu$.} ~ \\
\begin{itemize}
	\item \textbf{Alice:} Sample $x \in \set{0,1}^{E_0}$ uniformly at random and give the encoding graph $G_x$ to Alice. 
	\item \textbf{Bob:} Sample $i \in [k]$ uniformly at random and consider the matching $M_i$ of $G_x$. Give a perfect matching $M$ from $L \setminus L(M_i)$ to a new set of (server) vertices as 
	the input to Bob.  
\end{itemize}
\end{ourbox}

Given we have fixed the choice of the graph $G_0$, the following observation is immediate. 

\begin{observation}\label{obs:input-distribution}
	In the distribution $\mu$, the input to Alice is uniquely identified by $x \in \set{0,1}^{E_0}$ and the input to Bob is uniquely identified by $i \in [k]$. 
\end{observation}

\subsection{Analysis of the Input Distribution}

 Let $\pi$ be any \emph{deterministic} protocol for $\LB(n,\alpha)$ that succeeds with probability at least $\frac23$ on the inputs sampled from the distribution $\mu$. 

We use $\pi(x)$ to denote the message of Alice to Bob in the protocol (which by~\Cref{obs:input-distribution} is only a function of $x \in \set{0,1}^{E_0}$, hence the notation $\pi(x)$). 
Similarly, we use $a(\pi(x),i)$ to denote the \emph{assignment} output by Bob, given the message $\pi(x)$ and the index $i \in [k]$ as input (again, using~\Cref{obs:input-distribution}). 

In the following, we use $(X,I,\Pi)$ to denote, respectively, the \emph{random variable} for the input $x$ of Alice, the input $i$ of Bob, and the message $\pi(x)$ of Alice. 
We further use $X_I$ to denote the subsequence of $X$ that corresponds to the edges in $M_I$, namely, the ``special'' matching corresponding to Bob's input. 

Given that protocol $\pi$ is deterministic, the randomness of all these variables comes solely from the distribution $\mu$ of the inputs, namely, 
the choice of $(X,I) \sim \mu$. We prove that  protocol $\prot$ recovers a large fraction of the values in $X_I$ whenever it outputs a correct answer. 

\begin{lemma}\label{lem:matching-be-large}
	On any input $(x,i) \sim \mu$ that $\prot$ outputs a correct answer, at least half the values in $x_i$ are deterministically fixed given only $\pi(x)$ and $i$, where $x_i$ denotes the subsequence of $x$ corresponding to edges of $M_i$. 
\end{lemma}
\begin{proof}
	
	On any input $(x,i)$ sampled from $\mu$, there is an $L$-perfect matching: match $L \setminus L(M_i)$ using the new edges outside $G_x$ given to Bob and use edges of $M_i$ for assigning the  
	vertices in $L(M_i)$ to $R(M_i)$. This means that in any $\alpha$-approximate solution, load of any vertex in $R$ can be at most $\alpha$. 
	
	Let $A = a(\pi(x),i)$ be the assignment output by Bob. Consider vertices $S$ in $L(M_i)$ that do not use edges of $M_i$ in the assignment $A$. Since by~\Cref{obs:M0-M1}, $G_x$ is a $(2\alpha)$-\rsplus, all these vertices are incident on at most $\frac{1}{2\alpha} \cdot \card{M_i}$ 
	vertices $T \subseteq R$. 
Thus, the load of some vertex in $T$ is at least
	\[
		\frac{\card{S}}{\card{T}} \geq \frac{2 \alpha \cdot \card{S}}{\card{M_i}}.
	\]
Combining this with the upper bound of $\alpha$ on the load implies that $\card{S} \leq \frac12 \card{M_i}$. 
	This means that whenever Bob's output is correct, it contains at least $\frac12 \card{M_i}$ edges from $M_i$. 
        But each such edge uniquely identifies the corresponding bit in $x$ by the construction of $G_x$. 
	Thus, whenever the protocol is correct on input $(x,i)$, at least half the values of $x_i$ should be determined, 
given the assignment $a(\pi(x),i)$ output by Bob. 
\end{proof}

Using this lemma and the independence of the input to Alice and Bob (by~\Cref{obs:input-distribution}), we can conclude the proof using a (very) basic information theory argument. 
We first provide a brief refresher of basic information theory tools that we use in this proof. 

\subsubsection{A Quick Refresher on Information Theory}

 For a random variable $Y$ on support $\Omega$ with distribution $p(y)$ for $y \in \Omega$, the \textbf{entropy} of $Y$ is: 
\[
	\en{Y} := \sum_{y \in \Omega} p(y) \cdot \log{\frac{1}{p(y)}}. 
\]
The \textbf{conditional entropy} of a random variable $Y$ conditioned on another random variable $Z$ is:
\[
	\en{Y \mid Z} := \Exp_{z \sim Z}\bracket{\en{Y \mid Z=z}},
\]
where $\en{Y \mid Z=z}$ is the entropy of the random variable distributed as $p(y \mid Z=z)$ for $y \in \Omega$. 

We need the following facts about (conditional) entropy (the proofs of the statements in this fact are all standard applications of Jensen's inequality and can be found, e.g., in~\cite{CoverT06}). 

\begin{fact}[cf.~\cite{CoverT06}]\label{fact:info}
	For any random variables $Y,Z$ with support $\Omega$: 
	\begin{enumerate}
		\item\label{en:uniform} $0 \leq \en{Y} \leq \log\card{\Omega}$; moreover, the left (respectively, right) inequality is tight if and only if $Y$ is deterministic (respectively, uniformly distributed over $\Omega$); 
		\item\label{en:reduce} $\en{Y \mid Z} \leq \en{Y}$: conditioning (on a random variable) can only reduce the entropy;
		\item\label{en:chain-rule} $\en{Y,Z} = \en{Y} + \en{Z \mid Y}$: chain rule of entropy (also true for conditional entropy).  
	\end{enumerate}
\end{fact}

\subsubsection{Proof of~\Cref{thm:cc-lb}} 

By~\Cref{lem:matching-be-large}, with probability at least $\frac23$, Bob, given only $\Pi$ and $I \in [k]$ can recover at least half the values in $X_I$. This means there is ``considerably less'' uncertainty about $X_I$ conditioned on $\Pi$ and $I$, than without this conditioning. The following claim formalizes this. 

\begin{claim}\label{clm:info-step1}
	$\en{X_I \mid \Pi, I} \leq \frac89 \cdot r_0 + 1$. 
\end{claim}
\begin{proof}
Let $Z \in \set{0,1}$ be the indicator random variable for the event of~\Cref{lem:matching-be-large}, 
namely, $Z = 1$ iff the output by Bob is correct and  at least half the indices in $X_I$ are fixed by $\Pi$ and $I$.  Then
\begin{align*}
	\en{X_I \mid \Pi, I} &= \en{X_I,Z \mid \Pi, I} - \en{Z \mid X_I, \Pi,I} \tag{by chain rule in~\Cref{fact:info}-\eqref{en:chain-rule}} \\
	&= \en{X_I \mid \Pi,I,Z} + \en{Z \mid \Pi,I} - \en{Z \mid X_I, \Pi,I} \tag{again, by chain rule} \\
	&\leq \en{X_I \mid \Pi,I,Z}  + 1 \tag{as conditional entropy of $Z$ is non-negative and is at most $1$ by~\Cref{fact:info}-\eqref{en:uniform}} \\
	&= \Pr\paren{Z=1} \cdot \en{X_I \mid \Pi,I,Z=1} + \Pr\paren{Z=0} \cdot \en{X_I \mid \Pi,I,Z=0} + 1 \tag{by the definition of conditional entropy} \\
	&\leq \frac{2}{3} \cdot \frac{1}{2} \cdot \paren{\frac{4}{3} \cdot r_0} + \frac13 \cdot \paren{\frac{4}{3} \cdot r_0} + 1 \\
	&= \frac89 \cdot r_0 + 1. 
\end{align*}
where we used the following in the second to last step: the probability of $Z=1$ is at least $\frac23$ and conditioned on $Z=1$, $\Pi$ and $I$ reveal at least half of the indices of $X_I$ by~\Cref{lem:matching-be-large}, and so the 
	entropy of $X_I$ is at most $\frac12 \cdot \frac43 \cdot r_0$ by~\Cref{fact:info}-\eqref{en:uniform} since its unfixed part is a binary string of at most this length; the other term corresponding to $Z=0$ is bounded by~\Cref{fact:info}-\eqref{en:uniform} by simply
	using the fact that $X_I$ is a binary string of length at most $\frac43 \cdot r_0$ to begin with. 
\end{proof}

On the other hand, given that Alice is unaware of the choice of $I$, the only way for her to reduce the uncertainty about $X_I$ substantially, is to reduce the \emph{overall} uncertainty about her entire input. This in turn requires Alice to 
communicate a lot. The following claim formalizes this. 

\begin{claim}\label{clm:info-step2}
	$\en{X_I \mid \Pi,I} \geq  r_0 - \frac1k \cdot \norm{\prot}$. 
\end{claim}
\begin{proof}
By the definition of conditional entropy,
	\begin{align*}
		\en{X_I \mid \Pi, I} &= \sum_{i=1}^{k} \Pr\paren{I=i} \cdot \en{X_i \mid \Pi, I=i} \\
		&= \frac{1}{k} \cdot \sum_{i=1}^{k} \en{X_i \mid \Pi, I=i} \tag{as $I$ is distributed uniformly over $[k]$} \\
		&= \frac{1}{k} \cdot \sum_{i=1}^{k} \en{X_i \mid \Pi} \tag{as the joint distribution of $(X_i,\Pi=\Pi(X))$ is independent of the event $I=i$ by~\Cref{obs:input-distribution}} \\
		&\geq \frac{1}{k} \cdot \sum_{i=1}^{k} \en{X_i \mid \Pi, X_1,\ldots,X_{i-1}} \tag{as conditioning can only reduce the entropy by~\Cref{fact:info}-\eqref{en:reduce}} \\
		&= \frac{1}{k} \cdot \en{X \mid \Pi} \tag{by the chain rule of entropy in~\Cref{fact:info}-\eqref{en:chain-rule}} \\
		&\geq \frac{1}{k} \cdot \paren{\en{X} - \en{\Pi}} \tag{by applying chain rule and non-negativity of entropy} \\
		&\geq \frac{1}{k} \cdot \paren{m_0 - \norm{\prot}}, 
	\end{align*}
	where the last step holds because the distribution of $X$ is uniform over $2^{m_0}$ binary strings of length $m_0$ and the support of messages in $\prot$ are of size $2^{\norm{\prot}}$; hence, in both
	cases we can apply the inequality of~\Cref{fact:info}-\eqref{en:uniform} (which is tight for $\en{X}$, and provides an upper bound for $\en{\Pi}$). 
	Noting that $m_0 \geq k \cdot r_0$ concludes the proof. 
\end{proof}

We are ready to conclude the proof of~\Cref{thm:cc-lb}. 

\begin{proof}[Proof of~\Cref{thm:cc-lb}]
	
	The lower bound holds trivially whenever $\MC(n,4\alpha) = O(n\log{n})$ because even if the entire input is a random perfect matching given to Alice, she needs to communicate $\Omega(n\log{n})$ bits to send one edge
	per each vertex in $L$ to Bob which is needed for any finite approximation. Thus, in the following, we focus on the (only interesting) case when $\MC(n,4\alpha) = \omega(n\log{n})$, which implies that $r_0 = \omega(1)$. 

	By the easy direction of Yao's minimax principle, to prove the lower bound, we only need to focus on deterministic protocols that succeeds with probability at least $\frac23$ on inputs sampled from the distribution $\mu$ (this is simply an averaging argument over randomness of the protocol against the input distribution). For any such protocol $\prot$, by~\Cref{clm:info-step1} and~\Cref{clm:info-step2},  
	\[
		r_0 - \frac1k \cdot \norm{\prot} \leq \frac89 \cdot r_0 + 1, 
	\]
	which implies that 
	\[
		\norm{\prot} \geq \frac19 \cdot (k \cdot r_0) - k 
                   \gtrsim \frac{1}{\log{n}} \cdot \MC(n,4\alpha), 
	\]
	given that $r_0 = \omega(1)$. This concludes the proof. 
\end{proof}

\subsection{Proofs of~\Cref{cor:cc-lb,cor:stream-lb}}\label{sec:cor-lb}

We now provide short and standard proofs of~\Cref{cor:cc-lb,cor:stream-lb} using the results we established already. 

\begin{proofof}{\Cref{cor:cc-lb}}
	By~\Cref{thm:cc-lb}, (randomized) one-way communication complexity of $\alpha$-approximation of load-balancing can be lower bounded, up to a $\approx \log{n}$ term, by the density of $\Theta(\alpha)$-\rsplus s. 
	Setting $\alpha = n^{\frac14-O(\eps)}$ and using our construction of $(n^{\frac{1}{4}-O(\eps)})$-\rsplus s in~\Cref{thm:rsplus} with $n^{1+\Omega(\eps^2)}$ edges implies this corollary (note that the extra $\log{n}$ term 
	is subsumed by the hidden-constant of the $\Omega$-notation in the exponent). 
	
	The final part of the corollary holds by taking $\eps \rightarrow 0$ in the limit. 
\end{proofof}

\begin{proofof}{\Cref{cor:stream-lb}}
	We use the well-established fact that one-way communication complexity lower bounds imply space lower bounds for single-pass streaming algorithms. Given a streaming algorithm $A$ for load-balancing, we 
	obtain a one-way communication protocol as follows: Alice runs $A$ by treating her part of the input as the first part of the stream and then communicates the memory content of the algorithm to Bob, who continue running $A$ on his part of the input
	as the second part of the stream. This way, at the end, Bob obtains the output of $A$ on the entire input, with communication cost from Alice being at most equal to the worst-case memory size of the algorithm. 
	
	The lower bound for streaming algorithms now follows immediately from the above reduction and~\Cref{cor:cc-lb}. 
\end{proofof}

\subsection{Proof of~\Cref{thm:equivalence}: Sparsifiers = One-Way Communication}\label{sec:equivalence-proof} 

Finally, we provide the proof of our main equivalence result in~\Cref{thm:equivalence}, restated below. 

\begin{theorem*}[Restatement of~\Cref{thm:equivalence}]
Suppose there is a (randomized) communication protocol $\pi$ for $\LB(n,\alpha)$ with communication cost $\norm{\pi} \leq C$ and probability of success at least $2/3$. Then, 
\[
\spar(n,8\alpha) \lesssim C \cdot \log^2{(n)}.
\]
\end{theorem*}
\begin{proof}
	Suppose towards a contradiction that $\spar(n,8\alpha) \geq \eta \cdot C \cdot \log^2{n}$ for some sufficiently large constant $\eta$. Then, 
	\begin{itemize}
	\item By~\Cref{thm:MC-spar}, this implies that $\MC(G,4\alpha) \gtrsim \eta \cdot C \cdot \log{(n)}$;
	\item By~\Cref{thm:cc-lb}, this in turn implies that $\randC{\LB(n,\alpha)} \gtrsim \eta \cdot C$. 
	\item By taking $\eta$ to be a sufficiently large constant, this contradicts the fact that there is a randomized protocol $\pi$ for $\LB(n,\alpha)$ with $\norm{\pi} \leq C$. 
	\end{itemize}
	Thus, our contradicting assumption is false, and the theorem holds. 
\end{proof}

\section*{Acknowledgments}
We would like to thank Thatchaphol Saranurak for many fruitful discussions in the early stages of this project. 

Part of this work was conducted while the first named author was visiting the Simons
Institute for the Theory of Computing as part of the Sublinear Algorithms program.

\clearpage

\bibliographystyle{halpha-abbrv}
\bibliography{general}

\newcommand{\etalchar}[1]{$^{#1}$}
\begin{thebibliography}{HKPSR18}
\expandafter\ifx\csname url\endcsname\relax
  \def\url#1{\texttt{#1}}\fi
\expandafter\ifx\csname doi\endcsname\relax
  \def\doi#1{\burlalt{doi:#1}{http://dx.doi.org/#1}}\fi
\expandafter\ifx\csname urlprefix\endcsname\relax\def\urlprefix{URL }\fi
\expandafter\ifx\csname href\endcsname\relax
  \def\href#1#2{#2}\fi
\expandafter\ifx\csname burlalt\endcsname\relax
  \def\burlalt#1#2{\href{#2}{#1}}\fi

\bibitem[A24]{Assadi24}
S.~Assadi.
\newblock A simple {$(1-\eps)$}-approximation semi-streaming algorithm for
  maximum (weighted) matching.
\newblock In {\em Symposium on Simplicity in Algorithms, {SOSA} 2024}, pages
  337--354, 2024.

\bibitem[AB19]{AssadiB19}
S.~Assadi and A.~Bernstein.
\newblock Towards a unified theory of sparsification for matching problems.
\newblock In {\em 2nd Symposium on Simplicity in Algorithms, {SOSA} 2019,
  January 8-9, 2019}, volume~69, pages 11:1--11:20, 2019.

\bibitem[ABKL23]{AssadiBKL23}
S.~Assadi, S.~Behnezhad, S.~Khanna, and H.~Li.
\newblock On regularity lemma and barriers in streaming and dynamic matching.
\newblock In {\em {STOC} '23: 55th Annual {ACM} {SIGACT} Symposium on Theory of
  Computing}, 2023.

\bibitem[ABL20]{AssadiBL20}
S.~Assadi, A.~Bernstein, and Z.~Langley.
\newblock Improved bounds for distributed load balancing.
\newblock In {\em 34th International Symposium on Distributed Computing, {DISC}
  2020}, volume 179, pages 1:1--1:15, 2020.

\bibitem[ABL23]{AssadiBL23}
S.~Assadi, A.~Bernstein, and Z.~Langley.
\newblock All-norm load balancing in graph streams via the multiplicative
  weights update method.
\newblock In {\em 14th Innovations in Theoretical Computer Science Conference,
  {ITCS} 2023}, volume 251, pages 7:1--7:24, 2023.

\bibitem[AG11]{AhnG11}
K.~J. Ahn and S.~Guha.
\newblock Linear programming in the semi-streaming model with application to
  the maximum matching problem.
\newblock In {\em Automata, Languages and Programming - 38th International
  Colloquium, {ICALP} 2011}, pages 526--538, 2011.

\bibitem[AK24]{AssadiK24}
S.~Assadi and S.~Khanna.
\newblock Improved bounds for fully dynamic matching via {Ordered
  Ruzsa-Szemer\'edi} graphs.
\newblock {\em arXiv preprint arXiv:2406.13573}, 2024.

\bibitem[AKL16]{AssadiKL16}
S.~Assadi, S.~Khanna, and Y.~Li.
\newblock Tight bounds for single-pass streaming complexity of the set cover
  problem.
\newblock In {\em Proceedings of the 48th Annual {ACM} {SIGACT} Symposium on
  Theory of Computing, {STOC} 2016}, pages 698--711, 2016.

\bibitem[AKL17]{AssadiKL17}
S.~Assadi, S.~Khanna, and Y.~Li.
\newblock On estimating maximum matching size in graph streams.
\newblock In {\em Proceedings of the Twenty-Eighth Annual {ACM-SIAM} Symposium
  on Discrete Algorithms, {SODA} 2017}, pages 1723--1742, 2017.

\bibitem[AKLY16]{AssadiKLY16}
S.~Assadi, S.~Khanna, Y.~Li, and G.~Yaroslavtsev.
\newblock Maximum matchings in dynamic graph streams and the simultaneous
  communication model.
\newblock In {\em Proceedings of the Twenty-Seventh Annual {ACM-SIAM} Symposium
  on Discrete Algorithms, {SODA} 2016}, pages 1345--1364, 2016.

\bibitem[AKNS24]{AssadiKNS24}
S.~Assadi, C.~Konrad, K.~K. Naidu, and J.~Sundaresan.
\newblock {$O(\log\log{n})$} passes is optimal for semi-streaming maximal
  independent set.
\newblock In {\em Proceedings of the 56th Annual {ACM} Symposium on Theory of
  Computing, {STOC} 2024}, pages 847--858, 2024.

\bibitem[ALPZ21]{AhmadianLPZ21}
S.~Ahmadian, A.~Liu, B.~Peng, and M.~Zadimoghaddam.
\newblock Distributed load balancing: {A} new framework and improved
  guarantees.
\newblock In {\em 12th Innovations in Theoretical Computer Science Conference,
  {ITCS} 2021}, volume 185, pages 79:1--79:20, 2021.

\bibitem[AMS12]{AlonMS12}
N.~Alon, A.~Moitra, and B.~Sudakov.
\newblock Nearly complete graphs decomposable into large induced matchings and
  their applications.
\newblock In {\em Proceedings of the 44th Symposium on Theory of Computing
  Conference, {STOC} 2012}, pages 1079--1090, 2012.

\bibitem[AR20]{AssadiR20}
S.~Assadi and R.~Raz.
\newblock Near-quadratic lower bounds for two-pass graph streaming algorithms.
\newblock In {\em 61st {IEEE} Annual Symposium on Foundations of Computer
  Science, {FOCS} 2020}, pages 342--353, 2020.

\bibitem[AS23]{AssadiS23}
S.~Assadi and J.~Sundaresan.
\newblock Hidden permutations to the rescue: Multi-pass streaming lower bounds
  for approximate matchings.
\newblock In {\em 64th {IEEE} Annual Symposium on Foundations of Computer
  Science, {FOCS} 2023}, pages 909--932, 2023.

\bibitem[BCS74]{BrunoCS74}
J.~Bruno, E.~G. Coffman, Jr., and R.~Sethi.
\newblock Scheduling independent tasks to reduce mean finishing time.
\newblock {\em Comm. ACM}, 17:382--387, 1974.

\bibitem[Ber20]{Bernstein20}
A.~Bernstein.
\newblock Improved bounds for matching in random-order streams.
\newblock In {\em 47th International Colloquium on Automata, Languages, and
  Programming, {ICALP} 2020}, pages 12:1--12:13, 2020.

\bibitem[BG24]{BehnezhadG24}
S.~Behnezhad and A.~Ghafari.
\newblock Fully dynamic matching and {O}rdered {Ruzsa-Szemer\'edi} graphs.
\newblock {\em CoRR}, abs/2404.06069. To appear in FOCS 2024, 2024.

\bibitem[BO20]{BarenboimO20}
L.~Barenboim and G.~Oren.
\newblock Distributed backup placement in one round and its applications to
  maximum matching and self-stabilization.
\newblock In {\em Proc. 3rd Symposium on Simplicity in Algorithms}, pages
  99--105, 2020.

\bibitem[CHSW12]{CzygrinowHSW12}
A.~Czygrinow, M.~Hanćkowiak, E.~Szymańska, and W.~Wawrzyniak.
\newblock Distributed 2-approximation algorithm for the semi-matching problem.
\newblock In {\em Proc. 26th International Symposium on Distributed Computing},
  volume 7611, pages 210--222, 2012.

\bibitem[CKP{\etalchar{+}}21]{ChenKPSSY21}
L.~Chen, G.~Kol, D.~Paramonov, R.~R. Saxena, Z.~Song, and H.~Yu.
\newblock Almost optimal super-constant-pass streaming lower bounds for
  reachability.
\newblock In {\em {STOC} '21: 53rd Annual {ACM} {SIGACT} Symposium on Theory of
  Computing 2021}, pages 570--583, 2021.

\bibitem[DIMV14]{DemaineIMV14}
E.~D. Demaine, P.~Indyk, S.~Mahabadi, and A.~Vakilian.
\newblock On streaming and communication complexity of the set cover problem.
\newblock In {\em Distributed Computing - 28th International Symposium, {DISC}
  2014}, volume 8784, pages 484--498, 2014.

\bibitem[ER56]{ErdosR56}
P.~Erd\H{o}s and A.~R{\'e}nyi.
\newblock On some combinatorical problems.
\newblock {\em Publ. Math. Debrecen}, 4:398--405, 1956.

\bibitem[FKM{\etalchar{+}}05]{FeigenbaumKMSZ05}
J.~Feigenbaum, S.~Kannan, A.~McGregor, S.~Suri, and J.~Zhang.
\newblock On graph problems in a semi-streaming model.
\newblock {\em Theor. Comput. Sci.}, 348(2-3):207--216, 2005.

\bibitem[FLN{\etalchar{+}}02]{FischerLNRRS02}
E.~Fischer, E.~Lehman, I.~Newman, S.~Raskhodnikova, R.~Rubinfeld, and
  A.~Samorodnitsky.
\newblock Monotonicity testing over general poset domains.
\newblock In {\em Proceedings on 34th Annual {ACM} Symposium on Theory of
  Computing, 2002}, pages 474--483, 2002.

\bibitem[FLN14]{FakcharoenpholLN14}
J.~Fakcharoenphol, B.~Laekhanukit, and D.~Nanongkai.
\newblock Faster algorithms for semi-matching problems.
\newblock {\em ACM Trans. Algorithms}, 10(3):Art. 14,23, 2014.

\bibitem[FMU22]{FischerMU22}
M.~Fischer, S.~Mitrovic, and J.~Uitto.
\newblock Deterministic {$1+\eps$}-approximate maximum matching with
  {poly}$(1/\eps)$ passes in the semi-streaming model and beyond.
\newblock In {\em 54th Annual {ACM} {SIGACT} Symposium on Theory of Computing,
  2022}, pages 248--260, 2022.

\bibitem[FNSZ20]{FeldmanNSZ20}
M.~Feldman, A.~Norouzi{-}Fard, O.~Svensson, and R.~Zenklusen.
\newblock The one-way communication complexity of submodular maximization with
  applications to streaming and robustness.
\newblock In {\em Proceedings of the 52nd Annual {ACM} {SIGACT} Symposium on
  Theory of Computing, {STOC} 2020}, pages 1363--1374, 2020.

\bibitem[GKK12]{GoelKK12}
A.~Goel, M.~Kapralov, and S.~Khanna.
\newblock On the communication and streaming complexity of maximum bipartite
  matching.
\newblock In {\em Proceedings of the Twenty-Third Annual {ACM-SIAM} Symposium
  on Discrete Algorithms, {SODA} 2012}, pages 468--485, 2012.

\bibitem[Hal87]{Hall87}
P.~Hall.
\newblock On representatives of subsets.
\newblock {\em Classic Papers in Combinatorics}, pages 58--62, 1987.

\bibitem[HKPSR18]{HalldorssonKPR18}
M.~M. Halldórsson, S.~Köhler, B.~Patt-Shamir, and D.~Rawitz.
\newblock Distributed backup placement in networks.
\newblock {\em Distrib. Comput.}, 31(2):83--98, 2018.

\bibitem[HLLT06]{HarveyLLT06}
N.~J.~A. Harvey, R.~E. Ladner, L.~Lovász, and T.~Tamir.
\newblock Semi-matchings for bipartite graphs and load balancing.
\newblock {\em J. Algorithms}, 59(1):53--78, 2006.

\bibitem[HLW06]{HLW06}
S.~Hoory, N.~Linial, and A.~Wigderson.
\newblock Expander graphs and their applications.
\newblock {\em Bull. Amer. Math. Soc.}, 43:439--562, 2006.

\bibitem[Hor73]{Horn73}
W.~A. Horn.
\newblock Minimizing average flow time with parallel machines.
\newblock {\em Oper. Res.}, 21(3), 1973.

\bibitem[JR17]{JansenR17}
K.~Jansen and L.~Rohwedder.
\newblock On the configuration-lp of the restricted assignment problem.
\newblock In {\em Proceedings of the Twenty-Eighth Annual {ACM-SIAM} Symposium
  on Discrete Algorithms, {SODA} 2017}, pages 2670--2678, 2017.

\bibitem[JR20]{JansenR20}
K.~Jansen and L.~Rohwedder.
\newblock A quasi-polynomial approximation for the restricted assignment
  problem.
\newblock {\em {SIAM} J. Comput.}, 49(6):1083--1108, 2020.

\bibitem[Kap13]{Kapralov13}
M.~Kapralov.
\newblock Better bounds for matchings in the streaming model.
\newblock In {\em Proceedings of the Twenty-Fourth Annual {ACM-SIAM} Symposium
  on Discrete Algorithms, {SODA} 2013}, pages 1679--1697, 2013.

\bibitem[Kap21]{Kapralov21}
M.~Kapralov.
\newblock Space lower bounds for approximating maximum matching in the edge
  arrival model.
\newblock In D.~Marx, editor, {\em Proceedings of the 2021 {ACM-SIAM} Symposium
  on Discrete Algorithms, {SODA} 2021, Virtual Conference, January 10 - 13,
  2021}, pages 1874--1893. {SIAM}, 2021.

\bibitem[KKA23]{KhannaKA23}
S.~Khanna, C.~Konrad, and C.~Alexandru.
\newblock Set cover in the one-pass edge-arrival streaming model.
\newblock In {\em Proceedings of the 42nd {ACM} {SIGMOD-SIGACT-SIGAI} Symposium
  on Principles of Database Systems, {PODS} 2023}, pages 127--139, 2023.

\bibitem[KN97]{KNbook}
E.~Kushilevitz and N.~Nisan.
\newblock {\em Communication complexity}.
\newblock Cambridge University Press, 1997.

\bibitem[KR13a]{KonradR13}
C.~Konrad and A.~Ros{\'{e}}n.
\newblock Approximating semi-matchings in streaming and in two-party
  communication.
\newblock In {\em Automata, Languages, and Programming - 40th International
  Colloquium, {ICALP} 2013}, volume 7965, pages 637--649, 2013.

\bibitem[KR13b]{KonradR13arxiv}
C.~Konrad and A.~Ros{\'{e}}n.
\newblock Approximating semi-matchings in streaming and in two-party
  communication.
\newblock 2013, \burlalt{1304.6906}{http://arxiv.org/abs/1304.6906}.

\bibitem[LL04]{LinL04}
Y.~Lin and W.~Li.
\newblock Parallel machine scheduling of machine-dependent jobs with
  unit-length.
\newblock {\em European J. Oper. Res.}, 156(1):261--266, 2004.

\bibitem[LST90]{LenstraST90}
J.~K. Lenstra, D.~B. Shmoys, and {\'{E}}.~Tardos.
\newblock Approximation algorithms for scheduling unrelated parallel machines.
\newblock {\em Math. Program.}, 46:259--271, 1990.

\bibitem[MV17]{McGregorV17}
A.~McGregor and H.~T. Vu.
\newblock Better streaming algorithms for the maximum coverage problem.
\newblock In {\em 20th International Conference on Database Theory, {ICDT}
  2017}, volume~68, pages 22:1--22:18, 2017.

\bibitem[OBL18]{OrenBL18}
G.~Oren, L.~Barenboim, and H.~Levin.
\newblock Distributed fault-tolerant backup-placement in overloaded wireless
  sensor networks.
\newblock In {\em Proc. 9th International Conference on Broadband
  Communications, Networks, and Systems}, pages 212--224, 2018.

\bibitem[OZ22]{Olesker-TaylorZ22}
S.~Olesker{-}Taylor and L.~Zanetti.
\newblock Geometric bounds on the fastest mixing markov chain.
\newblock In {\em 13th Innovations in Theoretical Computer Science Conference,
  {ITCS} 2022}, volume 215, pages 109:1--109:1, 2022.

\bibitem[RS78]{RuzsaS78}
I.~Z. Ruzsa and E.~Szemer{\'e}di.
\newblock Triple systems with no six points carrying three triangles.
\newblock {\em Combinatorics (Keszthely, 1976), Coll. Math. Soc. J. Bolyai},
  18:939--945, 1978.

\bibitem[RY20]{RYbook}
A.~Rao and A.~Yehudayoff.
\newblock {\em Communication Complexity: and Applications}.
\newblock Cambridge University Press, 2020.

\bibitem[Yao79]{Yao79}
A.~C. Yao.
\newblock Some complexity questions related to distributive computing
  (preliminary report).
\newblock In {\em Proceedings of the 11h Annual {ACM} Symposium on Theory of
  Computing, 1979}, pages 209--213, 1979.

\end{thebibliography}

\clearpage
\appendix

\part*{Appendix}

\section{Deferred Proofs}\label{sec:deferred-proofs}

\subsection{Proof of \Cref{thm:spar-implies-protocol}}

\begin{proof}
	The protocol is simple:
	Alice computes an $\alpha$-approximation load-balancing sparsifier $H_A$ of $G_A$ with at most $T$ edges and sends $H_A$ to Bob.
	Each edge requires $O(\log(n))$ bits to send, for a total of $O(T\log(n))$ bits.
	
	Define $G \defeq G_A \cup G_B$ and $G' \defeq H_A \cup G_B$. To prove correctness, we need to show that $\optload(G') \leq \alpha \cdot \optload(G)$. Let $\Astar$ be the optimal assignment in $G$. 
	We partition $L$ into sets $L_A$ and $L_B$: for every $x \in L$, add $x$ to $L_A$ if $(x,\Astar(x)) \in G_A$ and add $x$ to $L_B$ if $(x,\Astar(x)) \in G_B$; if the edge is in both $G_A$ and $G_B$, then assign $x$ to $L_B$. 
	
	Let $\Astar_A$ denote the assignment $\Astar$ restricted to $L_A$ and define $\Astar_B$ analogously.
	Note that $\Astar_A$ is an assignment in $G_A$, $\Astar_B$ is an assignment in $G_B$, and $\load(\Astar_A),\load(\Astar_B)$ are both at most $\load(\Astar)$. 
	Since $\Astar_B$ is contained in $G_B \subseteq G'$, we can also use in our final assignment for $G'$. But we must replace $\Astar_A$ with a new assignment that is contained in $H_A$. To this end, note that by the existence of $\Astar_A$ we have $\optload(G_A[L_A \cup R]) \leq \load(\Astar_A) \leq \load(\Astar)$, so by definition of a load-balancing sparsifier, there exists an assignment $\A'_A$ of $H_A[L_A \cup R]$ with $\load(\A'_A) \leq \alpha \cdot \load(\Astar_A) \leq \alpha \cdot  \load(\Astar)$.
	
	We now define an assignment $\A'$ of $G'$ as follows: for $x \in L_A$ we set $\A'(x) = \A'_A(x)$ and for $x \in L_B$ we set $\A'(x) = \Astar_B(x)$. It is easy to see that all edges of $A'$ are contained in $G'$, and that $\load(\A') \leq \load(\A'_A) + \load(\Astar_B) \leq (\alpha + 1)\load(\Astar)$, as desired. 
\end{proof}

\label{sec:app-proofs}

\subsection{Reducing the Number of Servers}
\label{sec:reducing-num-servers}

In order to apply \Cref{lem:spar-LP} inside \Cref{thm:MC-spar}, we rely on the following claim, which shows that one can assume w.l.o.g that $\card{R} \leq \card{L}^2$.

\begin{claim}
	\label{claim:truncate-servers}
	Given a bipartite graph $G = (L, R,E)$ with $\card{L} = n$, there exists a subgraph $G' = (L, R',E')$ of $G$ such that $\card{R'} \leq n^2$ with $\spar(G,\alpha) \leq \spar(G',\alpha)$ and $\MC(G',\alpha) \leq \MC(G,\alpha)$ for any $\alpha \geq 1$.
\end{claim}

\begin{proof}
	Define a vertex $v \in L$ to be high-degree if $\deg_G(v) > n$.
	We define $E'$ as follows: start with $E' = E$, and then for every high-degree vertex $v \in L$, remove an aribtrary set of $\deg_G(V) - n$ edges incident to $v$. Every $v \in L$ now has $\deg_{E'}(v) \leq n$. Let $R'$ contain all vertices in $R$ with at least one incident edge in $E'$; it is easy to see that $\card{R'} \leq n^2$. We then define $G' := (L, R', E')$
	
	To show that $\spar(G,\alpha) \leq \spar(G',\alpha)$, we argue that any $\alpha$-sparsifier $H'$ of $G'$ is also an $\alpha$-sparsifier of $G$. 
	Using the criterion of load-balancing sparsifiers in Property (\ref{item:operational}) of \autoref{lem:equiv-sparsifiers}, it suffices to show that any set $X \subseteq L$ that is matchable in $G$ is also matchable in $G'$. Let $M$ be the matching from $X$ to $R$ in $G$; we argue that $X$ is also matchable in $G'$. 
	For every $u \in X$, if $u$ is not of high-degree in $G$, then $E_G(u) = E_{G'}(u)$, so we can use the same edge from $M$. 
	If $u$ is of high degree in $G$, then since $\deg_{G'}(u) = n$, there must be at least one free vertex in $R'$ that $u$ can be matched to.
	
	The inequality $\MC(G',\alpha) \leq \MC(G,\alpha)$ follows from the fact that $G'$ is a subgraph of $G$.
\end{proof}

\clearpage

\subsection{Proof of~\Cref{prop:set-family}}\label{app:prop-set-family}

\begin{proofof}{\autoref{prop:set-family}}
We use the probabilistic method. 
Suppose we pick $t$ random subsets $S\in {[2k]\choose k}$ independently and uniformly at random. 
For any two random subsets $S$ and $S'$, we would like to bound the probability that their intersection size is equal to $\ell$ for any $\frac{k}{2} \leq  \ell \leq k$.
If $S$ and $S'$ are chosen uniformly and independently, then 
    \begin{align*}
        \Pr[|S\cap S'| = \ell] = \frac{{k\choose \ell}{k\choose k-\ell}}{{2k \choose k}} = \frac{{k\choose k-\ell}^2}{{2k \choose k}}.
    \end{align*}
To see this, if we fix $S$, the numerator counts the number of subsets of size $k$ that intersects $S$ exactly $\ell$ times. 
This can be done by picking $\ell$ elements from $S$ and picking $k-\ell$ elements from $\overline{S}$, for both of which there are exactly ${k\choose \ell} = {k\choose k-\ell} $ ways. 
By the union bound, the probability that any two sets have intersection size greater than $k-s$ is
    \begin{align*}
        \Pr[|S\cap S'| \geq k-s] \leq \sum_{\ell=k-s}^k  \frac{{k\choose k-\ell}^2}{{2k \choose k}} = \sum_{\ell=0}^{s}  \frac{{k\choose \ell}^2}{{2k \choose k}} \leq s \cdot \frac{{k\choose s}^2}{{2k \choose k}}
    \end{align*}
    
Let $s=\delta k$ for some positive constant $\delta < \frac12$. 
Using Stirling's approximation $n! \asymp \sqrt{2\pi n}(n/e)^n$, 
    \begin{align*}
        s \cdot \frac{{k\choose s}^2}{{2k \choose k}} 
        &\asymp s\cdot \bigg(\frac{\sqrt{2\pi k}\cdot k^k}{2\pi \sqrt{s(k-s)}s^{s}(k-s)^{k-s}}\bigg)^2 \bigg/ \frac{\sqrt{4\pi k}(2k)^{2k}}{2\pi k\cdot  k^{2k}}\\
        &\asymp \frac{sk^{1.5}}{s(k-s)}\cdot \frac{k^{4k}}{(2k)^{2k}s^{2s}(k-s)^{2(k-s)}}\\
        &= \frac{ \sqrt{k}}{(1-\delta)}\cdot\frac{k^{4k}}{(2k)^{2k}(\delta k)^{2\delta k}(k(1-\delta))^{2k(1-\delta)}}\\
        &= \frac{\sqrt{k}}{(1-\delta)}\cdot\frac{1}{(4\cdot \delta^{2\delta}\cdot(1-\delta)^{2(1-\delta)})^k}\\
        &\leq  2\sqrt{k} \cdot c_\delta^{-2k}
    \end{align*}
    Finally, by using the union bound over all pairs of random subsets, we see that as long as $t \leq \frac{1}{2} k^{-\frac14} \cdot c_\delta^{k}$, then every pair of subsets has intersection size less than $(1 - \delta) k$ with positive probability. 
This implies that there exists a set family $\mathcal{F}$ as claimed in the statement.
\end{proofof}

\end{document}